\RequirePackage{fix-cm}
\documentclass[smallextended]{svjour3}       % onecolumn (second format)
\smartqed  % flush right qed marks, e.g. at end of proof

\usepackage{graphicx}
\usepackage{calligra,mathrsfs,mathtools}
\usepackage{bm,amsmath,amssymb}
\usepackage{bbm}
\usepackage{tikz-cd}
\usepackage{enumitem}

\usepackage{amsthm}
\newtheorem{thm}{Theorem}[section]
\numberwithin{thm}{section}% reset theorem numbering for each chapter

\theoremstyle{definition}
\newtheorem{defn}[thm]{Definition}% definition numbers are dependent on theorem numbers
\newtheorem{rmk}[thm]{Remark}

\newtheorem{lem}[thm]{Lemma}
\numberwithin{equation}{section}

\journalname{Bulletin of Mathematical Biology}
\begin{document}
 ~
\vspace{1cm}

\noindent {\Large \bf Algebraic Coarse-Graining of Biochemical \\ Reaction Networks}

\vspace{0.5cm}

\noindent {\small \bf Dimitri Loutchko \footnote{The University of Tokyo, Graduate School of Frontier Sciences, Department of Complexity Science and Engineering\\
	      5-1-5 Kashiwanoha, Kashiwa-shi, Chiba-ken 277-8561 \\
              \email{d.loutchko@edu.k.u-tokyo.ac.jp}}}

\vspace{3.5cm}

\title{Algebraic Coarse-Graining of Biochemical Reaction Networks%\thanks{Grants or other notes
}
%\subtitle{Construction and Application to Self-Sustaining Networks}


\author{Dimitri Loutchko  
}


\institute{D. Loutchko  \at
	      The University of Tokyo, Graduate School of Frontier Sciences, Department of Complexity Science and Engineering\\
	      5-1-5 Kashiwanoha, Kashiwa-shi, Chiba-ken 277-8561 \\
              \email{d.loutchko@edu.k.u-tokyo.ac.jp}              %  \\
}


%\maketitle

\begin{abstract}
Biological systems exhibit processes on a wide range of time and length scales.
This work demonstrates that models, wherein the interaction between system constituents is captured by algebraic operations, inherently allow for successive coarse-graining operations through quotients of the algebra.
Thereby, the class of model is retained and all possible coarse-graining operations are encoded in the lattice of congruences of the model.
We analyze a class of algebraic models generated by the subsequent and simultaneous catalytic functions of chemicals within a reaction network.
Our ansatz yields coarse-graining operations that cover the network with local functional patches and delete the information about the environment, and complementary operations that resolve only the large-scale functional structure of the network.
Finally, we present a geometric interpretation of the algebraic models through an analogy with classical models on vector fields.
We then use the geometric framework to show how a coarse-graining of the algebraic model naturally leads to a coarse-graining of the state-space.
The framework developed here is aimed at the study of the functional structure of cellular reaction networks spanning a wide range of scales.
\end{abstract}

\section*{Introduction}

rocesses in biology take place on many different time and length scales.
This is evident already at the level of single cells.
In the temporal domain, chemical reactions catalyzed by enzymes have characteristic scales that range from microseconds to seconds \cite{Winterbourn1975}, over hours for genetic regulation \cite{Rosenfeld2005} up to the order of days for the completion of the cell cycle \cite{Norbury1992}.
The complex organization of a living cell, however, gives rise to a situation, where even a single observable can exhibit fluctuations with a continuous $1/f$-spectrum.
Such spectra have been measured for glycolytic oscillations in yeast \cite{Aon2008}, in cardiac cells \cite{Rouke1994} and in electroencephalograms of human brains \cite{Pritchard1992}.

The spatial structures found in cells range from the subnanometer scale for small metabolites over proteins and complexes on the scale of tens of nanometers and larger to cell organelles measuring micrometers.
In the last decades, the viewpoint emerged that the cytoplasm is a highly structured functional unit.
Experiments have shown that transient formation of protein complexes occurs frequently and often entire metabolic pathways \cite{An2008} and cell signaling cascades \cite{Good2011, Haga2012} are carried out within such complexes under limited exchange of matter with the environment \cite{Miles1999}.
Globally, the cytoplasm mediates strongly non-local effects of cyclic conformational molecular motions on metabolic diffusivity \cite{Mikhailov2015, Jee2018} and exhibits glass-like properties impacting all intracellular processes involving large components \cite{Parry2014}, providing hints at the large scale spatial and temporal structure of the cytoplasm.

From these experimental results the viewpoint emerges that the scales occurring in cellular processes do not possess a discrete spectrum, but that it is dense in both the spatial and temporal domain.
In particular, this implies that a change of scale of a model via coarse-graining based on scale separation might not be, even in principle, possible for models of complex biological systems.
With a point of view moving towards systems biology, aiming at a holistic description of biological systems, this can pose a serious obstacle. 

In this work, we present an approach that circumvents the difficulty of directly coarse-graining the state space in order to achieve a scale-transformation by coarse-graining the space of functions acting on the state space.
The advantage of this approach is that the space of functions is endowed with a natural algebraic structure, which descends to the quotients of the functional algebra and thus to the corase-grained models.
This means that the possible coarse-graining procedures are encoded in the lattice of congruences of the functional algebra and that consecutive coarse-graining procedures using increasingly coarse congruences lead to a multiscale description of the system.

We demonstrate this idea on a class of algebraic models for biochemical reaction networks.
These models are based on the chemical reaction system (CRS) formalism developed by Hordijk and Steel \cite{Hordijk2004}.
The main application for CRS has been the extensive and successful study of self-sustaining reaction networks \cite{Steel2000,Mossel2005,Hordijk2014,Smith2014}.
In \cite{Loutchko2019}, it was shown that the CRS formalism has a natural algebraic structure corresponding to subsequent and simultaneous catalytic events and the respective algebraic models were constructed.

In section \ref{sec:SGM}, the semigroup models and their basic properties are reviewed.
Section \ref{sec:ACG} expands the idea of algebraic coarse-graining sketched above and presents a class of congruences that correspond to coarse-graining of the small-scale structure of the network and a complementary class that corresponds to coarse-graining of the environment.
Finally, in section \ref{sec:Geometric}, which is the core of this article, we add another layer to the formalism by establishing a ``geometric'' viewpoint of the semigroup models.
This approach is motivated by a correspondence to classical models that employ vector fields.
It is shown how the algebraic models are attached to the state space, which is the power set of all chemicals of the network under consideration, in a compatible manner.
Thereby, the dynamics is given by a section compatible with the partial orders on state space and on the functional algebra.
We show how the coarse-graining procedure by a congruence on the functional algebra descends to the state space and thereby leaves all structures and compatibilities intact.
As a demonstration, we discuss the geometry of the congruences introduced in section \ref{sec:ACG}.
All references to the Supplementary Information are denoted by a capital S.
The mathematical background needed for this work is covered in section S1.

\section{Semigroup Models of CRS} \label{sec:SGM}

\subsection*{The formalism of CRS}

This introduction to the chemical reaction system (CRS) formalism follows \cite{Hordijk2004}.
A classical chemical reaction network (CRN) is a finite set of chemicals $X$ together with a set of reactions $R=\{r_i\}_{i \in I}$ indexed by a finite set $I$ each equipped with a reaction rate constant.
A reaction $r \in R$ is usually written as
\begin{equation} \label{eq:reaction}
 a_1 A_1 + a_2 A_2 + ... + a_n A_n \longrightarrow b_1 B_1 + b_2 B_2 + ... + b_m B_m ,
\end{equation}
\noindent where $a_i,b_j \in \mathbb{N}$ and $A_i,B_j \in X$, $A_i \neq B_j$ for $i=1,...,n$ and $j=1,...,m$.
We will only utilize the sets of substrates and products, which we call the domain $ \text{dom}(r) = \{A_1,...A_n\} $ and range $\text{ran}(r)=\{B_1,...,B_m\}$ of a reaction $r$ given by \ref{eq:reaction}, because the CRS formalism does not employ detailed kinetic information, but instead emphasizes the catalytic function of the chemicals in $X$.

\begin{defn} \label{def:CRS}
A {\it chemical reaction system} (CRS) is a triple $(X,R,C)$, where $X$ is a finite discrete set of chemicals, $R$ is a finite set of reactions and $C \subset X \times R$ is a set of reactions catalyzed by chemicals of $X$.
For any pair $(x,r) \in C$, the reaction $r$ is said to be catalyzed by $x$.
\end{defn}

Following \cite{Bonchev2012} a CRS can be graphically represented by a graph with two kinds of vertices and two kinds of directed edges.
As an example, consider the graph in Fig. \ref{fig:example}.
The solid disks correspond to the chemicals $X$ and the open circles correspond to the reactions from $R$.
The chemicals participating in a reaction are shown by solid arrows.
If a reaction is catalyzed by some chemical, this is indicated by a dashed arrow.

\begin{figure}[htb]
  \centering
  \includegraphics[width=0.35\linewidth]{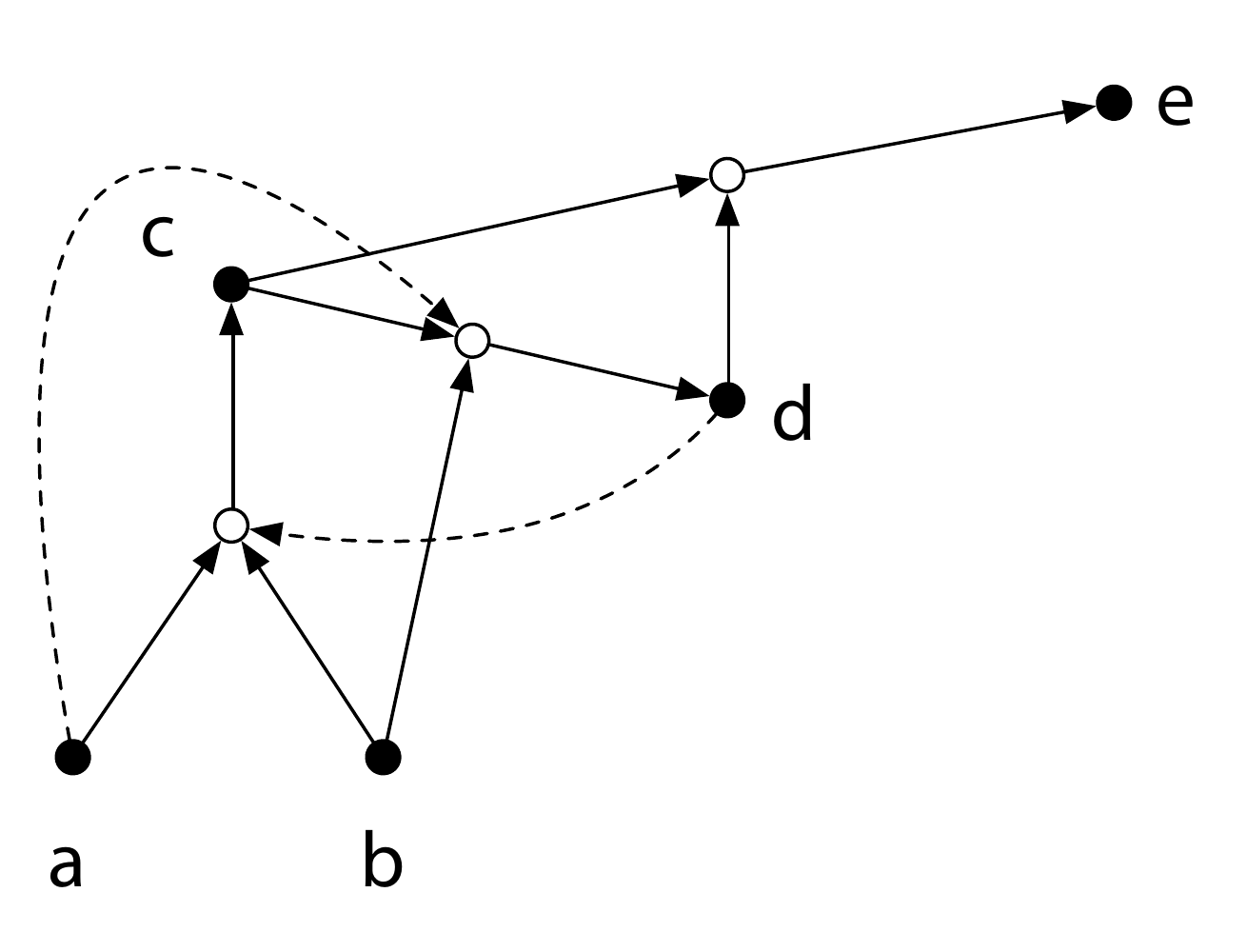}
  \caption[Example of semigroup models with food set]{
  Example of a graphical representation of a CRS.
  The CRS consists of five chemicals $X=\{a,b,c,d,e\}$ and three reactions $a+b\rightarrow c$, $c+b \rightarrow d$ and $c+d \rightarrow e$.
  The first two reactions are catalyzed by $d$ and $a$, respectively, whereas the last reaction is not catalyzed.
  }
  \label{fig:example}
\end{figure}

The catalytic function of chemicals can be equipped with a natural algebraic structure, namely, the subsequent function and the simultaneous function, as well as combinations thereof.

\subsection*{The algebraic structure of a CRS}

Throughout this section, let $(X,R,C)$ be a CRS.
The state of the CRS is defined by the presence or absence of the chemicals, i.e. by giving the subset $Y \subset X$ of chemicals that are present.
Thus the state space of the CRS is the power set $\mathfrak{X} = \{0,1\}^X$.

A reasonable way to define the function of some given chemical $x \in X$ is via the reactions it catalyzes, i.e. by the way it acts on the state space $\mathfrak{X}$.
This definition is motivated by the work of Rhodes \cite{Rhodes2010}.

\begin{defn} \label{def:func}
Let $(X,R,C)$ be a CRS with state space $\mathfrak{X} = \{0,1\}^X$.
The {\it function} $\phi_r : \mathfrak{X} \rightarrow \mathfrak{X}$ of a reaction $r \in R$ is defined as
\begin{equation*}
  \phi_r(Y) =\begin{cases}
    \text{ran}(r) & \text{if $\text{dom}(r) \subset Y$}\\
	\emptyset & \text{else}
  \end{cases}
\end{equation*}

\noindent for all $Y \subset X$.
The sum $\phi + \psi$ of two functions $\phi, \psi: \mathfrak{X} \rightarrow \mathfrak{X}$ is given by
\begin{equation} \label{eq:addition}
  (\phi + \psi)(Y) = \phi(Y) \cup \psi(Y)
\end{equation}

\noindent for all $Y \subset X$.
The {\it function} $\phi_x: \mathfrak{X} \rightarrow \mathfrak{X}$ of $x \in X$ is defined as the sum over all reactions catalyzed by $x$
\begin{equation*}
\phi_x = \sum_{(x,r) \in C} \phi_r.
\end{equation*}

\end{defn}

\noindent Two functions $\phi_x$ and $\phi_y$ with $x,y \in X$ can be composed via 
\begin{equation*}
(\phi_x \circ \phi_y)(Y) := \phi_x (\phi_y(Y)) \text{ for any $Y \subset X$}.
\end{equation*}

\noindent The composition $\circ$ is the usual composition of maps and therefore associative.
The addition is extended to arbitrary functions via the formula (\ref{eq:addition}).
It is associative, commutative and idempotent (cf. S1.9).
This leads to the definition of the functional algebra of a CRS.

\begin{defn}

Let $(X,R,C)$ be a CRS.
Its {\it functional algebra} $(\mathcal{S}(X),\circ,+)$ is the smallest subalgebra of the full algebra of functions $\mathcal{T}(\mathfrak{X})$ (definition S1.19)
that contains $\{\phi_x\}_{x \in X}$ and the zero function given by $0(Y) = \emptyset$ for all $Y \subset X$ and is closed under the operations $\circ$ and $+$ . 
We denote this algebra by
\begin{equation*}
\mathcal{S}(X) = \langle \phi_x \rangle_{x \in X}.
\end{equation*}

\noindent Analogously, for any subset of chemicals $Y \subset X$, the subalgebra $\mathcal{S}(Y)$ of $\mathcal{S}(X)$ of functions supported on $Y$ is defined as $\mathcal{S}(Y) = \langle \phi_x \rangle_{x \in Y}$ and $\mathcal{S}(\emptyset) = \{0\}$ is the trivial algebra.
\end{defn}

\begin{defn}
When we consider only the multiplicative structure on $\mathcal{S}(X)$ or only the additive structure, the resulting objects $(\mathcal{S}(X),\circ)$ and $(\mathcal{S}(X),+)$, are semigroups (cf. S1.20-S1.22).
\end{defn}

As a subalgebra of $\mathcal{T}(\mathfrak{X})$, $\mathcal{S}(X)$ is {\it finite}.
The two operations $\circ$ and $+$ have obvious interpretations in terms of the function of enzymes on a CRS:
The sum of two functions $\phi_x + \phi_y$, $x,y \in X$ describes the {\it joint} or {\it simultaneous} function of two enzymes $x$ and $y$ on the network - it captures the reactions catalyzed by both $x$ and $y$ at the same time.
The composition of two functions $\phi_x \circ \phi_y$, $x,y \in X$ describes the {\it subsequent} function on the network: first $y$ and then $x$ act by their respective catalytic function.
By definition $\mathcal{S}(X)$ captures all possibilities of joint and subsequent functions of elements of the network on the network itself.
The following properties follow directly from the definitions.

\begin{lem} \label{lemma:properties}

Let $\mathcal{S}(X)$ be the algebra of functions of the CRS $(X,R,C)$.

\noindent (I) There is a natural partial order on $\mathcal{S}(X)$ given by 
\begin{equation} \label{eq:PO}
 \phi \leq \psi \Leftrightarrow \phi(Y) \subset \psi(Y) \text{ for all $Y \subset X$}.
\end{equation}

\noindent (II) Any $\phi, \psi \in \mathcal{S}(X)$ satisfy 
\begin{equation} \label{eq:order1}
 \phi \leq \phi + \psi.
\end{equation}

\end{lem}

To facilitate the discussion in the following sections, we give an explicit representation of the elements of $\mathcal{S}(X)$.

\begin{lem}[\cite{Loutchko2019}, lemma 4.1] \label{lemma:function}

Any element $\phi \in \mathcal{S}(X)$ can be written as a nested sum
\begin{equation} \label{eq:phi2}
\phi = \sum_{y_1 \in Y} \phi_{y_1} \circ (\sum_{y_2 \in Y_{y_1}} \phi_{y_2} \circ ( ... \circ (\sum_{y_n \in Y_{y_1y_2...y_{n-1}}} \phi_{y_n})...))
\end{equation}

\noindent for some $n \in \mathbb{N}$, where $Y,Y_{y_1} ...,Y_{y_1y_2...y_{n-1}}$ are {\it multisets} (possibly empty) of elements in $X$ and each $Y_{y_1y_2...y_j}$ with $j < n$ depends on the elements $y_1 \in Y, y_2 \in  Y_{y_1},...,y_j \in Y_{y_1y_2...y_{j-1}}$.

\end{lem}

\begin{rmk} \label{rmk:function}
The previous lemma implies that each element $\phi \in \mathcal{S}(X)$ can be represented as a tree with edges labeled by functions $\phi_y$ and the vertices representing sums over the underlying edges.
The sums are then multiplied with the function on the edge above the respective vertex.
Fig. \ref{fig:tree}A gives an example of such a representation.\\

\begin{figure}[ht]
  \centering
  \includegraphics[width=.7\linewidth]{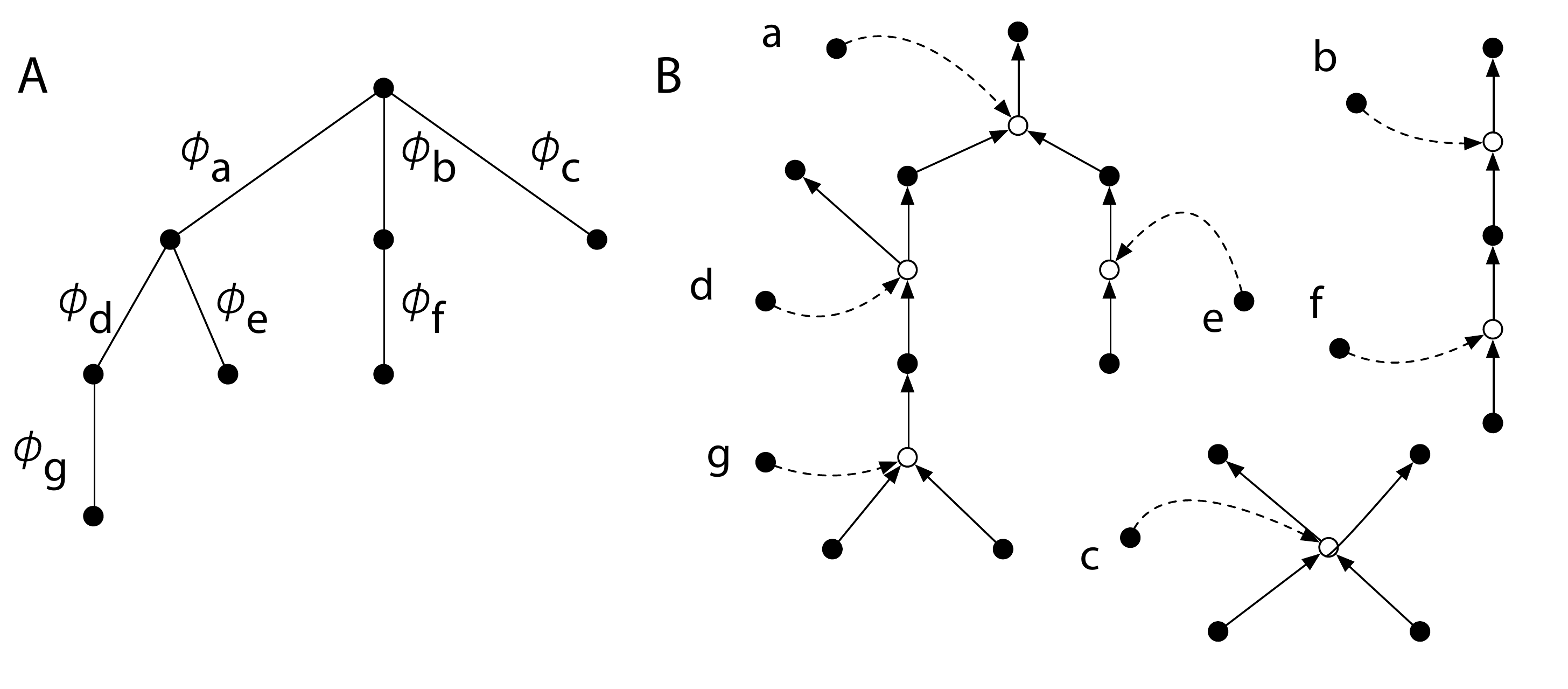}
  \caption[Example of a representation of a function as a tree]{The tree {\bf A} shows the function $\phi = \phi_a \circ ((\phi_d \circ \phi_g) + \phi_e) + (\phi_b \circ \phi_f) + \phi_c$ as an example of an explicit representation of a general element of $\mathcal{S}$ as discussed in the text.
  {\bf B} visualizes the reaction pathway within a CRS corresponding to the function represented in {\bf A}.
  As the root of the tree {\bf A} has three branches, the pathway has three components that are not interconnected.
  Note that the pathway {\bf B} does not represent a unique function.
  For example, it is also the pathway corresponding to the function $\phi + \phi_g$.}
  \label{fig:tree}
\end{figure}

The representation of a function by a tree implies a correspondence to reaction pathways in the CRS, where the leafs of the tree correspond to starting reactions and vertices correspond to joining reaction pathways.
As an example, Fig. \ref{fig:tree}B shows the pathways corresponding to the tree from Fig. \ref{fig:tree}A.
However, the mapping of functions to reaction pathways is not always injective.
For example, the reaction pathway shown in Fig. \ref{fig:tree}B corresponds to the function $\phi$ represented in Fig. \ref{fig:tree}A, but it is also the reaction pathway of the function $\phi + \phi_g$.

\end{rmk}

\section{Algebraic Coarse-Graining} \label{sec:ACG}

\subsection*{General approach} \label{sec:introCongruences}

The following considerations apply to any algebra in the sense of universal algebra, but for clarity we restrict ourselves to the semigroup $(\mathcal{S}(X),+)$.
The idea of functional coarse-graining is to determine all partitions $\rho$ of the set $\mathcal{S}(X)$ such that the algebraic operation descends from $\mathcal{S}(X)$ to operations between the sets of the partition.
This means that the functions of $\mathcal{S}(X)$ should be grouped into classes that behave ``similarly'' with respect to the algebraic operation.
The advantage of this procedure is that the reduced, i.e. coarse-grained, space of functions tautologically has the same algebraic operation as the original model and therefore retains the same biological interpretation.
Thus the class of models is not changed and then further coarse-graining can be iteratively performed in the same manner to obtain a description of the system on many scales.
The lattice of congruences of the algebra characterizes all possibilities for such coarse-graining procedures and moreover is endowed with a partial order that precisely determines the possibility of consecutive coarse-graining procedures.

Let us now formulate the above discussion in mathematical terms (cf. section S2 for details and definitions).
Being a partition of $\mathcal{S}(X)$ means that $\rho$ is an {\it equivalence relation}.
We write $\phi \rho \psi$ if and only if $\phi$ and $\psi$ are in the same equivalence class.
The set of all equivalence classes is denoted by $\mathcal{S}(X)/\rho$ and the equivalence class of $\phi \in \mathcal{S}(X)$ is denoted by $\phi \rho$.
For the descent of the algebraic operation from $\mathcal{S}(X)$ to $\mathcal{S}(X)/\rho$ to be well-defined, the relation $\rho$ must be a congruence, i.e. satisfy
\begin{equation} \label{eq:cong}
 \phi \rho \psi \textrm{ and } \phi' \rho \psi' \Rightarrow (\phi + \phi') \rho (\psi + \psi')
\end{equation}

\noindent for all $\phi, \phi', \psi, \psi' \in \mathcal{S}(X)$.
Then the operation $+$ on $\mathcal{S}(X)/\rho$ is independent of the choice of equivalence class representatives.

Due to property \ref{eq:cong}, all properties of the operation $+$ on $\mathcal{S}(X)$ are inherited by $+$ on $\mathcal{S}(X)/\rho$ and thus $\mathcal{S}(X)/\rho$ becomes a semigroup. 
It is called the {\it quotient} of $(\mathcal{S}(X),+)$ by $\rho$.
Analogously, if $\rho$ satisfies $\phi \rho \psi \textrm{ and } \phi' \rho \psi' \Rightarrow (\phi \circ \phi') \rho (\psi \circ \psi')$ for all $\phi, \phi', \psi, \psi' \in \mathcal{S}(X)$, then it is a congruence on $(\mathcal{S}(X),\circ)$.
If $\rho$ is a congruence on $(\mathcal{S}(X),+)$ and $(\mathcal{S}(X),\circ)$, then it is a congruence on $(\mathcal{S}(X),\circ,+)$.
We note that the lattice of congruences of $(\mathcal{S}(X),\circ,+)$ is a sublattice of both the lattices of congruences of $(\mathcal{S}(X),\circ)$ and $(\mathcal{S}(X),+)$ and thus it can be studied by via the lattices on $(\mathcal{S}(X),\circ)$ and $(\mathcal{S}(X),+)$ individually.

\subsection*{Congruences on Semigroup models}

This section focuses on congruences on $(\mathcal{S}(X),+)$.
Coarse-graining procedures for the functions of $(\mathcal{S}(X),+)$ corresponding to small pathways via remark \ref{rmk:function} as well as large pathways and combinations thereof are presented.
Congruences on $(\mathcal{S}(X),\circ)$ are treated in section S3.
Finally, the natural inverse to coarse-graining via semigroup extensions is discussed.

\subsubsection*{Congruences on $(\mathcal{S}(X),+)$}

The length $len(\phi)$ of a function $\phi \in \mathcal{S}(X)$ captures the size of the pathway corresponding to $\phi$ via remark \ref{rmk:function} and is defined as follows.

\begin{defn} \label{def:len}
 For any $\phi \in \mathcal{S}(X)$, let $len(\phi)$ be the largest integer $n$ such that there is a non-zero function $\psi \in \mathcal{S}(X)^n$ that satisfies $\psi \leq \phi$.
 Hereby, $\mathcal{S}(X)^n$ is the ideal of $\mathcal{S}(X)$ consisting of all elements of the form $a_1 \circ a_2 \circ ... \circ a_n$ for $a_1, ..., a_n \in \mathcal{S}(X)$.
 The length of the zero function is $0$.
\end{defn}

By this definition $len$ satisfies
\begin{equation} \label{eq:len}
 len(\phi + \psi) = \max \{ len(\phi), len(\psi) \}
\end{equation}

\noindent for any $\phi, \psi \in \mathcal{S}(X)$.
Here, the inequality $len(\phi + \psi) \geq \max \{ len(\phi), len(\psi) \}$ follows from lemma \ref{lemma:properties}(II).
The opposite inequality follows from the fact that the sum of two functions cannot have a longer pathway of subsequent catalyzed reactions than the ones already contained within one of summands.

This leads to a definition of some special equivalence relations $\rho_n$ on $\mathcal{S}(X)$ for any $n \in \mathbb{N}$ by stipulating that the functions $\phi, \psi \in \mathcal{S}(X)$ are in the same equivalence class if and only if their lengths do not exceed $n$, i.e.
\begin{equation} \label{eq:rho_n}
 \phi \rho_n \psi \Leftrightarrow len(\phi) \leq n \textrm{ and }len(\psi) \leq n
\end{equation}

\noindent in addition to $\phi \rho_n \phi$ for all $\phi \in \mathcal{S}(X)$.
Equation \ref{eq:len} immediately implies that the relations $\rho_n$ satisfy the property \ref{eq:cong} and are thus congruences.
Moreover, the $\rho_n$ are totally ordered by inclusion as subsets of $\mathcal{S}(X) \times \mathcal{S}(X)$
\begin{equation*}
 \rho_0 \subsetneq \rho_1 \subsetneq ... \subsetneq \rho_N
\end{equation*}

\noindent and $\rho_N = \mathcal{S}(X) \times \mathcal{S}(X)$ for $N = \max_{\phi \in \mathcal{S}(X)} \{len(\phi)\}$.
This gives rise to projections
\begin{equation*}
 \pi_{n,k}: \mathcal{S}(X) / \rho_{n} \twoheadrightarrow \mathcal{S}(X) / \rho_{n+k}
\end{equation*}

\noindent for $0 \leq n \leq N$ and $0<k\leq N-n$.
The $\pi_{n,k}$ are naturally homomorphisms with respect to addition.
Moreover, there are inclusions
\begin{equation*}
 \iota_{n,k}: \mathcal{S}(X) / \rho_{n+k} \xhookrightarrow{} \mathcal{S}(X) / \rho_n
\end{equation*}

\noindent for $0 \leq n \leq N$ and $0<k\leq N-n$.
{\it A priori}, the $\iota_{n,k}$ are just maps of sets.
In general, the obstruction for being homomorphisms is that there can exist $\phi, \psi \in \mathcal{S}(X) / \rho_{n+k}$ such that $\phi + \psi = 0$ in $\mathcal{S}(X) / \rho_{n+k}$, but $\iota_{n,k}(\phi) + \iota_{n,k}(\psi) \neq 0$ in $\mathcal{S}(X) / \rho_n$.
However, the property \ref{eq:len} ensures that this does not occur and thus $\iota_{n,k}$ are homomorphisms with respect to addition.

A biological interpretation of the quotients $\mathcal{S}(X) / \rho_n$ and the maps $\pi_{n,k}$ and $\iota_{n,k}$ now follows from remark \ref{rmk:function}.
The elements of $\mathcal{S}(X) / \rho_n$ corresponding to pathways of length less or equal to $n$ are set to $0$ and are therefore not resolved anymore.
The non-zero elements of $\mathcal{S}(X) / \rho_n$ capture the global structure of the network and contain only pathways that have sufficient length.
Therefore, taking the quotient of $\mathcal{S}(X)$ with respect to $\rho_n$ corresponds to extraction of the larger {\it functional structure} and neglecting the functional structure up to a given size $n$.
For $n=0$, the whole functionality is resolved.
With increasing $n$, more and more functions disappear until all functions are set to $0$ for $n=N$ and the quotient $\mathcal{S}(X) / \rho_N$ becomes trivial.
Thereby, the projections $\pi_{n,k}$ correspond to the negligence of functions with length between $n$ and $n+k$ and the $\iota_{n,k}$ correspond to the inclusion of functions of length at least $n+k$ into $\mathcal{S}(X) / \rho_n$.

\subsection*{The inverse of coarse-graining}
For $0 \leq k \leq N$, define the congruences $\rho^k$ on $\mathcal{S}(X)$ by
\begin{equation*}
 \phi \rho^k \psi \Leftrightarrow len(\phi) \geq k \textrm{ and }len(\psi) \geq k.
\end{equation*}

\noindent and $\phi \rho^k \phi$ for all $\phi \in \mathcal{S}(X)$.
$\rho^k$ groups all functions of length at least $k$ into one equivalence class and fully resolves all smaller functions.
This corresponds to a coarse-graining of the environment around patches of functions shorter than $k$.
Again, relation \ref{eq:len} ensures that $\rho^k$ is a congruence with respect to addition.
For $k,n$ such that $N \geq k > n \geq 0$, we can define the congruences
\begin{equation*}
 \rho^k_n = \rho^k \cup \rho_n
\end{equation*}

\noindent corresponding to a resolution of functions $\phi \in \mathcal{S}(X)$ with $n \leq len(\phi) < k$ and to the coarse-graining of all shorter functions and all longer functions into single equivalence classes in the quotient $\mathcal{S}(X)/\rho^k_n$.
Note that this is an instance of lattice algebra (S1.11) and thus such construction can be carried further using arbitrary combinations of the lattice operations.
The inclusion $\iota^k_n :  \mathcal{S}(X) / \rho^k_n \xhookrightarrow{} \mathcal{S}(X) / \rho_n$ is a semigroup homomorphism and we have a short exact sequence of semigroups (cf. S1.2)
\begin{equation*}
 0 \rightarrow \mathcal{S}(X)/\rho^k_n \xrightarrow{\iota^k_n} \mathcal{S}(X) / \rho_n \xrightarrow{\pi_{n,k-n}} \mathcal{S}(X) / \rho_{k} \rightarrow 0,
\end{equation*}

\noindent i.e. the semigroup $\mathcal{S}(X) / \rho_n$ is an extension of $\mathcal{S}(X) / \rho_{k}$ by $\mathcal{S}(X)/\rho^k_n$.
This is {\it verbatim} the biological interpretation: the functions of length at least $k$ extended by the functions of length between $k$ and $n$ give functions of length at least $n$.
While $\mathcal{S}(X) / \rho_{k}$ encodes the large scale functional structure of the network and $\mathcal{S}(X)/\rho^k_n$ a strictly lower scale, $\mathcal{S}(X) / \rho_n$ resolves the functions on both scales.
This suggests that within the considered algebraic framework the inverse procedure to coarse-graining is given by extensions of the algebra.
%For example, in the case of groups it is known that the isomorphism classes of central extensions of a group $G$ by a group $A$ are one-to-one correspondence with the elements of the cohomology group $\text{H}^2(G,A)$.
Under the given biological interpretation, the study of semigroup extensions becomes the study of the possibilities to couple two systems on different scales in an algebraically consistent way.
For the semigroup models, a solid theoretical basis in terms of the generalization of the homological characterizations of group extensions to semigroups is already known \cite{Wells1978}.

\section{A geometric point of view} \label{sec:Geometric}

The congruences $\rho_n$, $\rho^k$, $\rho^k_n$ and their combinations through the lattice operations join and meet allow to construct coarse-graining procedures more easily and flexibly compared to the techniques employed by the classical methods.
Moreover, the kinds of coarse-grained structures obtained often go beyond what is possible with classical approaches.
For example, the congruences $\rho^k$ correspond to a coarse-graining of the environment.
Classically, one would fix some subsystem and integrate over the other degrees of freedom, i.e. over the environment, to obtain a coarse-grained description.
In contrast, the elements of $\mathcal{S}(X)/\rho^k$ still resolve the full network structure, but a given class of functions is ``integrated out''. 
Whereas the classical procedure is based on a reduction of the system's {\it state space}, the algebraic procedure is a reduction of the system's {\it functional space}.
In this section, a geometric picture is developed, wherein the algebraic models are the functional algebras of the state space $\mathfrak{X} = \{0,1\}^X$.
From this geometric point of view, we show that there is a natural way for coarse-graining of the state-space resulting from coarse-graining of $(\mathcal{S}(X),+)$ by any congruence.

\subsubsection*{Dynamics on a semigroup model}

In \cite{Loutchko2018}, a discrete dynamics on $\mathfrak{X}$ is introduced to characterize self-sustaining chemical reaction systems.
Thereby, for each set $Y \subset X$, its {\it function} $\Phi_Y:\mathfrak{X} \rightarrow \mathfrak{X}$ is defined as
\begin{equation} \label{eq:PhiY}
\Phi_Y = \sum_{\phi \in \mathcal{S}(Y)} \phi.
\end{equation}

\noindent Equivalently, $\Phi_Y$ is the unique maximal element of $\mathcal{S}(Y)$.
If $Y \subset Z \subset X$, then lemma \ref{lemma:properties}(II) implies $\Phi_Y \leq \Phi_Z$.

\begin{defn}
The {\it discrete dynamics} on a CRS $(X,R,C)$ with initial condition $Y_0 \subset X$ is generated recursively by the propagator $\mathcal{D}: \mathfrak{X} \rightarrow \mathfrak{X}$ via $Y \mapsto \Phi_Y(Y)$.
Analogously, the dynamics can be parametrized by $\mathbb{N}$ as $Y_{n+1} = \Phi_{Y_n}(Y_n)$ for all $n \in \mathbb{N}$.
\end{defn}

\subsection*{Geometry of semigroup models}

Classically, the dynamics of a chemical reaction network is modeled on the state space $\mathbb{R}^N_{\geq 0}$, which keeps track of the exact concentrations of the chemicals from $X$, where $N=|X|$ is assumed to be finite.
The time evolution of the concentrations through chemical reactions is governed by a set of differential equations $dx/dt = f_{\text{class}}(x)$, which are usually derived from mass action kinetics.
In the parlance of differential geometry, the whole system is described by the real manifold $\mathbb{R}^N_{\geq 0}$ with a smooth section $f_{\text{class}}$ into its tangent bundle
\[
  \begin{tikzcd}
    T\mathbb{R}^N_{\geq 0}  \arrow[swap]{d}{\pi} & \\
     \mathbb{R}^N_{\geq 0}. \arrow[swap,bend right]{u}{f_{\text{class}}} &
  \end{tikzcd}
\]

The physical system is modeled by an appropriate choice of the initial condition $x_0 \in \mathbb{R}^N_{\geq 0}$ and by integration of the differential equation $dx/dt = f_{\text{class}}(x)$.
This yields a trajectory $x(t)$ parametrized by the semigroup $(\mathbb{R}_{\geq 0},+)$.

Analogously, one can view the algebra of functions $\mathcal{S}(X)$ as a structure $\mathcal{S}$ above the state space $\mathfrak{X}$
\[
  \begin{tikzcd}
    \mathcal{S}  \arrow[swap]{d}{\pi} & \\
     \mathfrak{X}, \arrow[swap,bend right]{u}{f} &
  \end{tikzcd}
\]

\noindent where $\mathcal{S}$ is defined in analogy to the tangent bundle as $\mathcal{S} = \coprod_{Y \in \mathfrak{X}} \mathcal{S}(Y)$, such that the fiber over $Y \subset X$ is $\mathcal{S}(Y)$.
The partial order on $\mathfrak{X}$ given by inclusion of sets induces a compatible partial order on the fibers via $\mathcal{S}(Y') < \mathcal{S}(Y)$ for $Y' \subset Y$.
Denote by $\iota_{Y,Y'}: \mathcal{S}(Y') \lhook\joinrel\longrightarrow \mathcal{S}(Y)$ the inclusion of subalgebras.
The partial order and this compatibility is the analogue of matching Euclidean the topologies of $\mathbb{R}^N_{\geq 0}$ and its tangent bundle.
The dynamics on $\mathfrak{X}$ is determined by a choice of elements $f_Y \in \mathcal{S}(Y)$ for each $Y$.
The smoothness of the section $f$ in the classical case is reflected by the compatibility of the partial order on $\mathfrak{X}$ with the partial order on the semigroup elements, i.e.
\begin{equation} \label{eq:dynamics_compatibility}
 \iota_{Y,Y'}(f_{Y'}) \leq f_Y
\end{equation}

\noindent for $Y' \subset Y$.
The trajectory for any initial condition $Y_0 \subset X$ is parametrized by $\mathbb{N}$ via $Y_{n+1} = f_{Y_n}(Y_n)$.
One natural example of a choice of dynamics is $f_Y = \Phi_Y$ by definition \ref{eq:PhiY}.
This geometric viewpoint is illustrated in the left part of Fig. \ref{fig:geometry}.
All details of this analogy to classical dynamics governed by a vector field are shown in Table S1.

\begin{figure*}[h]
  \centering
  \includegraphics[width=1.0\linewidth]{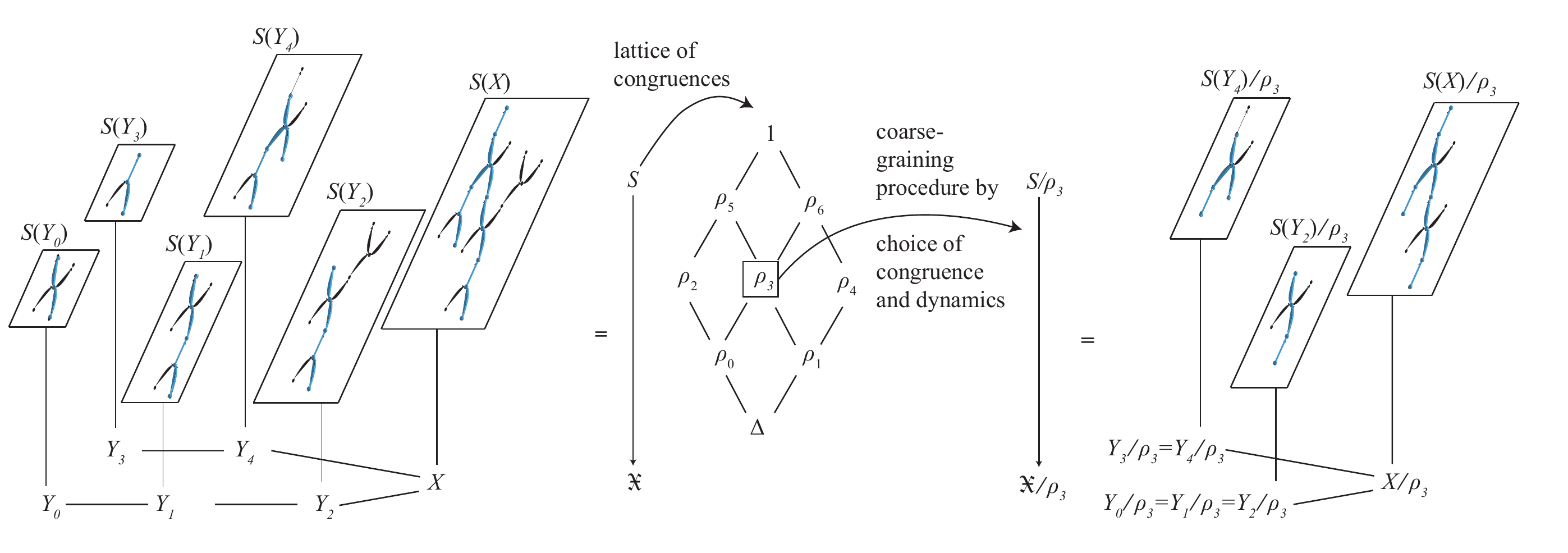}
  \caption{
  A sample representation of the geometric object $\mathcal{S} \rightarrow \mathfrak{X}$.
  The partial order on $\mathfrak{X}$ is visualized in the bottom part of the graph.
  Over each $Y \in \mathfrak{X}$, the algebra of functions $\mathcal{S}(Y)$ is represented by a network.
  Note, however, that according remark \ref{rmk:function}, the algebra contains several network functions.
  The blue part of the network represents the section $f: \mathfrak{X} \rightarrow \mathcal{S}$.
  In the middle, a representation of the lattice of congruences on $\mathcal{S}(X)$ is shown.
  Picking a congruence leads to a coarse-graining of the geometric space, whereby certain functional algebras and certain elements of $\mathfrak{X}$ belong to the same partitions and thus are considered as identical in the coarse-grained description.
  }
  \label{fig:geometry}
\end{figure*}

\subsection*{Coarse-graining in the geometric context}

Each point of the state space $\mathfrak{X}$ has an algebra of functions $\mathcal{S}(Y)$, which is a subalgebra of $\mathcal{S}(X)$, attached to it.
A congruence $\rho$ on $\mathcal{S}(X)$ descends to a congruence on $\mathcal{S}(Y)$, which is also denoted by $\rho$.
This is compatible with the inclusion maps $\iota_{Y,Y'}$, which descend to $\iota_{Y,Y'}^{\rho}: \mathcal{S}(Y')/\rho \hookrightarrow \mathcal{S}(Y)/\rho$ for $Y' \subset Y$.
Thereby, a projection $\mathcal{S} = \coprod_{Y \in \mathfrak{X}} \mathcal{S}(Y) \longrightarrow \coprod_{Y \in \mathfrak{X}} \mathcal{S}(Y)/\rho$ to a space over $\mathfrak{X}$ is induced.
This shows the coarse-grained {\it functional structure} over $\mathfrak{X}$ and suggests to group elements of the state space according to their function.
Note that the functions in the $\mathcal{S}(Y)/\rho$ are not well-defined on $\mathfrak{X}$.
However, it is possible to construct a natural equivalence relation $\rho_X$ on $\mathfrak{X}$ such that the functions in $\mathcal{S}(Y)/\rho$ are well-defined on the quotient $\mathfrak{X}/\rho_X$. 
In this regard, define an equivalence relation $\rho_X^{\text{pre}}$ on $\mathfrak{X}$ via
\begin{equation} \label{eq:rho'}
 Y \rho_X^{\text{pre}} Y' \Leftrightarrow \iota_{X,Y}^{\rho}(\mathcal{S}(Y)/\rho) = \iota_{X,Y'}^{\rho}(\mathcal{S}(Y')/\rho).
\end{equation}

\noindent Note that actual equality and not just isomorphism is required.
The following partition of $\mathfrak{X}$ is the closest one to $\mathfrak{X}/\rho_X^{\text{pre}}$ among those with the property that the dynamics $f: \mathfrak{X} \rightarrow \mathcal{S}$ descends in a well-defined manner.

\begin{defn} \label{def:CGstateSpace}
 Let $\rho_X$ be the finest partition of $\mathfrak{X}$ such that $\rho_X^{\text{pre}} \leq \rho_X$ and for each $Y\rho_X \in \mathfrak{X}/\rho_X$, the dynamics
 \begin{equation} \label{eq:CGfunc}
  f^{\rho}_{Y\rho_X}(Y\rho_X) := g_Y(Y)\rho_X,
 \end{equation}

 \noindent is independent of the choice of coset representative of $Y\rho_X$ and of the choice of $g_Y$, which is any coset representative of $f_Y\rho$.
 
\end{defn}

This yields the desired coarse-grained geometric object
\begin{equation}\label{eq:CGgeometry}
  \begin{tikzcd}
    \mathcal{S}/\rho  \arrow[swap]{d}{\pi} & \\
     \mathfrak{X}/\rho_X \arrow[swap,bend right]{u}{f^{\rho}} &
  \end{tikzcd}
\end{equation}

\noindent with fibers $(\mathcal{S}/\rho)(Y\rho_X) = \coprod_{Y'\rho_X Y}\mathcal{S}(Y')/\rho$.
The coarse-graining approach presented above is illustrated in Fig. \ref{fig:geometry}.
%The well-definedness and compatibility of the partial orders on the structures are proven in section S5.

In section S4, we work out the geometric coarse-graining by the congruences $\rho_n$ defined in equation \ref{eq:rho_n}.
Essentially, depending on the size of $n$ and the architecture of the network, there are three qualitatively different cases.
1) $\rho = \rho_X^{\text{pre}}$: The functions and states set on $0$ and $\emptyset$ play no role in the larger-scale structure of the network.
2) $\rho > \rho_X^{\text{pre}}$: Some of the functions of length at most $n$ induce larger functionality.
3) $\rho = \mathfrak{X} \times \mathfrak{X}$: The geometric model is trivial and functions of length at most $n$ suffice to produce the whole system.
This example demonstrates how congruences as simple as the $\rho_n$ highlight functional aspects of biochemical reaction networks.
For small $n$, it is to be expected that the geometric coarse-graining procedure follows case 1) and leads with increasing $n$ via case 2) to case 3).
It is certainly interesting to study the quantitative changes of this behavior in biological reaction networks and compare the results to corresponding results on random networks.

\section{Discussion} \label{sec:discussion}

The main goal of this paper is to present a systematic algebraic approach to coarse-grain biological systems via their functionality.
The technical background behind this approach, besides possible philosophical or aesthetical considerations, is that the function of system components on each other inherently should have an algebraic structure.
Then the requirement for a functionally consistent corase-graining is equivalent to taking quotients of the algebra of functions.
We have worked out this approach for the functional structure of chemical reaction systems and have shown that it naturally implies a coarse-graining of the state space using a geometrically minded interpretation of the algebraic models.

The presented formalism is mathematically strict and constructive.
In addition, all structures are finite and thus all presented methods can be directly implemented as a program and applied to experimental biological data.
In comparison to classical models, the dimension of the state space in the algebraic models is much smaller.

Therefore, applications to systems with enough functionality including metabolites, enzymes, DNA, RNA and signaling molecules are conceivable.
Such large and functionally rich systems are precisely the target systems for the presented models.
It would be very interesting to determine and analyze the lattice of congruences of a large cellular network.
One could verify the commonly used functional partitions and interactions between cellular organelles and other components based on the lattice of congruences.
It will be exciting to actually see cell organelles automatically emerge after the algebraic course-graining of the state space.
With the presented methods, one could also attempt to recover the causality implied by the central dogma of biology and study possible obstructions to it.
More importantly, a wealth of new functional relationships - even between large-scale structures - could be found and aid in the discovery of new pharmaceutical applications.
In addition, the comparison of the statistical properties of the lattice of congruences of a biological system to those of random networks could provide new insight on the functional organization in biology, including the organization on large scales.

Finally, we note that the CRS formalism has already been applied to macroscopic problems such as ecology \cite{Gatti2018} and economy \cite{Hordijk2017} and thus our coarse-graining approach can also be applied to study modularity in these systems.

\begin{acknowledgements}
I am deeply indebted to Gerhard Ertl for valuable discussions and his enormous moral as well as financial support at the FHI in Berlin.
I am very thankful to Hiroshi Kori for stimulating discussions and his generous support at the University of Tokyo.
I thank J\"urgen Jost and Peter F. Stadler for discussing this work.
\end{acknowledgements}

\bibliographystyle{spphys}       % APS-like style for physics
\bibliography{literatur}   % name your BibTeX data base

%\begin{thebibliography}{}

\end{document}

~
\vspace{1cm}

\noindent {\Large \bf Supplementary Information: \\ Algebraic Coarse-Graining of Biochemical \\ Reaction Networks}

\vspace{0.5cm}

\noindent {\small \bf Dimitri Loutchko \footnote{The University of Tokyo, Graduate School of Frontier Sciences, Department of Complexity Science and Engineering\\
	      5-1-5 Kashiwanoha, Kashiwa-shi, Chiba-ken 277-8561 \\
              \email{d.loutchko@edu.k.u-tokyo.ac.jp}}}

\vspace{3.5cm}

\title{Algebraic Coarse-Graining of Biochemical Reaction Networks%\thanks{Grants or other notes
}
%\subtitle{Construction and Application to Self-Sustaining Networks}


\author{Dimitri Loutchko  
}


\institute{D. Loutchko  \at
	      The University of Tokyo, Graduate School of Frontier Sciences, Department of Complexity Science and Engineering\\
	      5-1-5 Kashiwanoha, Kashiwa-shi, Chiba-ken 277-8561 \\
              \email{d.loutchko@edu.k.u-tokyo.ac.jp}              %  \\
}

\section{Basic concepts and definitions}

The definitions given here follow \cite{Almeida1995} and \cite{Howie1995}.

\subsection{Algebraic objects}

In this section, an algebra is defined in the sense of universal algebra and related elementary concepts are presented.
Then the notions are specialized by application to the main objects encountered in the text, i.e. the algebra of functions $(\mathcal{S}(X),\circ,+)$ and its subalgebras $(\mathcal{S}(Y),\circ,+)$, the semigroups $(\mathcal{S}(X),+)$ and $(\mathcal{S}(X),\circ)$ and their subsemigroups as well as the lattice $\mathfrak{X}=\{0,1\}^X$.

\begin{defn}
 An {\it algebraic type} $\tau = (\mathcal{O},\alpha)$ is a pair, where $\mathcal{O}$ is the set of operations and $\alpha$ is a map $\alpha:\mathcal{O} \rightarrow \mathbb{N}$.
 For each $f \in \mathcal{O}$, we say that $\alpha(f) \in \mathbb{N}$ is the arity of the operation $f$.
\end{defn}

\begin{defn}
 An {\it algebra} $\mathfrak{A} = (\mathcal{A};F)$ of algebraic type $\tau=(\mathcal{O},\alpha)$ is a non-empty set $A$ and a function $F$ on $\mathcal{O}$ such that $F(f):\mathcal{A}^n\rightarrow \mathcal{A}$ with $n=\alpha(f)$ for all $f \in \mathcal{O}$.
 If $\mathcal{O}=\{f_1,...,f_k\}$ is finite, we write $\mathfrak{A} = ( \mathcal{A};f_1,...,f_k)$ and say that $\mathfrak{A}$ is of type $(n_1,...,n_k)$, where $n_i = \alpha(f_i)$ is the arity of the respective operation for $i=1,...,k$.
\end{defn}

\begin{rmk}
 Note that the notion of algebra used here is the notion employed in the area universal algebra and is different from the more commonly used notion of algebra over a ring in the area of commutative algebra.
 Examples for the latter are matrix algebras or polynomials over a commutative ring.
\end{rmk}

\begin{defn}
 A {\it direct sum} of finitely many algebras $\{\mathfrak{A}_i = (\mathcal{A}_i;F_i) \}_{i=1}^n$ of the same type $\tau$ is an algebra $\mathfrak{A} = \bigoplus_{i =1}^n\mathfrak{A}_i = (\bigoplus_{i=1}^n \mathcal{A}_i;\bigoplus_{i=1}^n F_i) $ of type $\tau$, where $\bigoplus_{i=1}^n \mathcal{A}_i$ is the direct sum of sets and the operations $\bigoplus_{i=1}^n F_i(f)$ are defined componentwise.
\end{defn}

\begin{defn}
 A {\it semigroup} is an algebra $(\mathcal{S};\circ )$ of type $(2)$ such that the operation $\circ$ is associative, i.e. $a \circ (b \circ c) = (a \circ b) \circ c$ for all $a,b,c \in \mathcal{S}$.
\end{defn}

\begin{ex} \label{ex:full_transformation_semigroup}
 For a finite set $A$, the {\it full transformation semigroup} $(\mathcal{T}(A),\circ)$ is the set of all maps $\{f:A \rightarrow A\}$.
 The semigroup operation is the composition of maps, i.e. $(f \circ g)(a) = f(g(a))$ for all $a \in A$.
\end{ex}

\begin{defn}
 A {\it semigroup with zero} is a semigroup $(\mathcal{S};\circ)$ with an element $0 \in \mathcal{S}$ such that $a \circ 0 = 0 \circ a = 0$ for all $a \in \mathcal{S}$.
 A semigroup with zero is often regarded as an algebra $(\mathcal{S};\circ,0)$ of type $(2,0)$.
\end{defn}

\begin{defn}
 A {\it commutative semigroup} is a semigroup $(\mathcal{S};\circ)$ such that $a \circ b = b \circ a$ for all $a,b \in \mathcal{S}$.
\end{defn}

\begin{defn} \label{def:idempotent}
 A {\it commutative semigroup of idempotents} is a commutative semigroup $(\mathcal{S};\circ)$ such that $a \circ a = a$ for all $a \in \mathcal{S}$.
\end{defn}

\begin{ex}
 A {\it direct sum} of two semigroups $(\mathcal{S}_1;\circ )$ and $(\mathcal{S}_2;\circ )$ is the semigroup $(\mathcal{S}_1 \oplus \mathcal{S}_2;\circ )$, where $(a_1,b_1) \circ  (a_2,b_2) = (a_1 \circ a_2, b_1 \circ b_2)$ for all $(a_1,b_1), (a_2,b_2) \in \mathcal{S}_1 \oplus \mathcal{S}_2$.
 In what follows, we omit the information of the operation $\circ$ and write $\mathcal{S}_1 \oplus \mathcal{S}_2$ for the direct sum.
\end{ex}

Now we give a central definition that will be used further in section \ref{sec:cong}.

\begin{defn}
A {\it lattice} is an algebra $(\mathcal{L},\lor,\land)$ of type $(2,2)$ such that the operations satisfy
 \begin{align*}
  x \vee (y \vee z) = (x \vee y) \vee z&; x \wedge (y \wedge z) = (x \wedge y) \wedge z \\
  x \vee y = y \vee x&; x \wedge y = y \wedge x\\
  x \vee x = x &; x \wedge x = x \\
  x \wedge (x \vee y) = x &; x \vee (x \wedge y) = x
 \end{align*}
 
 \noindent for all $x,y \in \mathcal{L}$.
 The operation $\vee$ is called {\it join} and $\wedge$ is referred to as {\it meet}.
\end{defn}

The following proposition gives an equivalent characterization of a lattice as a partially ordered set.

\begin{prop}[cf. \cite{Almeida1995},Prop.1.1.11.] \label{prop:cong}

Let $(\mathcal{L},\lor,\land)$ be a lattice.
Then
\begin{equation*}
 x \leq y \text{ iff } x = x \wedge y
\end{equation*}

\noindent defines a partial order on $\mathcal{L}$ such that
\begin{align*}
 \inf\{x,y\} &= x \vee y \\
 \sup\{x,y\} &= x \wedge y.
\end{align*}

\noindent Conversely, if $(\mathcal{L},\leq)$ is a partially ordered set such that 
\begin{align*}
 x \vee y &:= \inf\{x,y\} \\
 x \wedge y &:= \sup\{x,y\}
\end{align*}

\noindent exist for all $x,y \in \mathcal{L}$, then $(\mathcal{L},\lor,\land)$ is a lattice.
\end{prop}

\begin{ex}
 For any set $X$, the power set $\mathfrak{X}=\{0,1\}^X$ is a lattice.
 The partial order is given by inclusion of sets and the join and meet are given by the union and intersection of sets, respectively, i.e. 
 \begin{align*}
  Y \vee Y' = Y \cup Y',\\
  Y \wedge Y' = Y \cap Y'
 \end{align*}
 
 \noindent for all $Y,Y' \subset X$. 
\end{ex}

\begin{defn} \label{def:alg_hom}
 An {\it algebra homomorphism} from an algebra $(\mathcal{A};F)$ to an algebra $(\mathcal{B};G)$ of the same type $\tau = (\mathcal{O},\alpha)$ is a map $\phi:\mathcal{A} \rightarrow \mathcal{B}$ such that for all $f \in \mathcal{O}$ and all $a_1,...,a_{k_f} \in \mathcal{A}$
 \begin{equation*}
  \phi(F(f)(a_1,...a_{k_f}))=G(f)(\phi(a_1),...\phi(a_{k_f})),
 \end{equation*}
 
\noindent where $k_f$ is the arity of $f$.
\end{defn}

\begin{defn} \label{def:alg_iso}
 An {\it algebra isomorphism} from an algebra $(\mathcal{A};F)$ to $(\mathcal{B};G)$ of the same type $\tau$ is a homomorphism
 \begin{equation*}
  \phi:\mathcal{A} \rightarrow \mathcal{B}
 \end{equation*}

 \noindent that is one-to-one.
 If for any two algebras $(\mathcal{A};F)$ and $(\mathcal{B};G)$, there exists an isomorphism, we say that the algebras are {\it isomorphic} and write
 \begin{equation*}
  \mathfrak{A} \simeq \mathfrak{B}.
 \end{equation*}
\end{defn}

\begin{ex}
 A {\it semigroup homomorphism} is an algebra homomorphism from $(\mathcal{S};\circ)$ to $(\mathcal{T};\circ)$, i.e. it is a map $f:\mathcal{S} \rightarrow \mathcal{T}$ such that $f(a \circ b) = f(a) \circ f(b)$ for all $a,b \in \mathcal{S}$
\end{ex}

\begin{ex}
 A {\it homomorphism of semigroups with zero} from $(\mathcal{S};\circ,0)$ to $(\mathcal{T};\circ,0)$ is a semigroup homomorphism $f:\mathcal{S} \rightarrow \mathcal{T}$ such that $f(0)=0$.
\end{ex}

\begin{defn}
 A {\it subalgebra} of $\mathfrak{A} = (\mathcal{A};F_{\mathcal{A}})$ is an algebra $\mathfrak{B} = (\mathcal{B};F_{\mathcal{B}})$ of the same type $\tau=(\mathcal{O},\alpha)$ such that $\mathcal{B} \subset \mathcal{A}$ and $F_{\mathcal{B}}(f)$ is the restriction of $F_{\mathcal{A}}(f)$ from $\mathcal{A}^{\alpha(f)}$ to $\mathcal{B}^{\alpha(f)}$ for all $f \in \mathcal{O}$.
 
 \noindent More naturally, a {\it subalgebra} of $\mathfrak{A} = (\mathcal{A};F_{\mathcal{A}})$ is an algebra $\mathfrak{B} = (\mathcal{B};F_{\mathcal{B}})$ of the same type such that there exists an injective algebra homomorphism
 \begin{equation*}
  \iota: \mathcal{B} \rightarrow \mathcal{A}.  
 \end{equation*}

\end{defn}

We arrive at the {\bf main definition} of this section.

\begin{defn}
 The {\it full algebra of functions} $(\mathcal{T}(\mathfrak{X}),\circ,+,0)$ on a finite set $X$ is an algebra of type $(2,2,0)$, where $(\mathcal{T}(\mathfrak{X}),\circ)$ is the full transformation semigroup on the power set $\mathfrak{X}=\{0,1\}^X$ (cf. example \ref{ex:full_transformation_semigroup}).
 The operation of addition $+$ is defined as
 \begin{equation*}
  (f+g)(Y) = f(Y) \cup g(Y)
 \end{equation*}
 
 \noindent for all $Y \subset X$ and all $f,g \in \mathcal{T}(\mathfrak{X})$.
 Note that $+$ is commutative and idempotent (cf. definition \ref{def:idempotent}).
 The zero element $0$ is the constant map $0(Y) = \empty$ for all $Y \subset X$.
 It is the neutral element with respect to addition, i.e. $f+0 = 0+f =0$ for all $f \in \mathcal{T}(\mathfrak{X})$ and a left-zero with respect to multiplication, i.e. $0 \circ f = 0$ for all $f \in \mathcal{T}(\mathfrak{X})$.
 Whenever we consider the operation $+$ on $(\mathcal{T}(\mathfrak{X}),+)$, we do not explicitly mention the information on the zero element, but implicitly assume its existence.
\end{defn}

\begin{rmk}
 The algebra $(\mathcal{T}(\mathfrak{X}),\circ)$ obtained from $(\mathcal{T}(\mathfrak{X}),\circ,+,0)$ by discarding the operation $+$ and the information on $0$ is the full transformation semigroup on $\mathfrak{X}$ from example \ref{ex:full_transformation_semigroup}.
\end{rmk}

\begin{rmk}
 The algebra $(\mathcal{T}(\mathfrak{X}),+)$ obtained from $(\mathcal{T}(\mathfrak{X}),\circ,+,0)$ by discarding the operation $\circ$ and not explicitly showing the information on $0$ is a commutative semigroup of idempotents with zero.
\end{rmk}

\begin{ex}
 The {\it algebra of functions} $(\mathcal{S}(X),\circ,+)$ defined in the main text is a subalgebra of $(\mathcal{T}(\mathfrak{X}),\circ,+,0)$.
 The {\it multiplicative semigroup model} $(\mathcal{S}(X),\circ)$ is a subsemigroup of $(\mathcal{T}(\mathfrak{X}),\circ)$ and the {\it additive semigroup model}  $(\mathcal{S}(X),+)=(\mathcal{S}(X),+,0)$ is a subsemigroup of $(\mathcal{T}(\mathfrak{X}),+)$.
\end{ex}

\begin{rmk}
 For any $Y\subset X$, the algebra $(\mathcal{S}(Y),\circ,+)$ is defined as a subalgebra of $(\mathcal{T}(\mathfrak{X}),\circ,+,0)$, and not as a subalgebra of $(\mathcal{T}(\mathfrak{Y}),\circ,+,0)$ with $\mathfrak{Y}=\{0,1\}^Y$.
 Note that we follow the notation from the main text and omit $0$ from the notation $(\mathcal{S}(Y),\circ,+)$.
 By definition, $\mathcal{S}(Y)$ is generated by the functions supported on $Y$, i.e. by the set $\{\phi_x\}_{x \in Y}$.
 For any $Y' \subset Y \subset X$, the inclusion of sets 
 \begin{equation*}
  \iota^{Set}_{Y,Y'}: Y' \xhookrightarrow{} Y
 \end{equation*}

 \noindent induces an inclusion of functions
 \begin{equation*}
  \iota^{Functions}_{Y,Y'}: \{\phi_x\}_{x \in Y'} \xhookrightarrow{} \{\phi_x\}_{x \in Y},
 \end{equation*}
 
 \noindent which induces an algebra homomorphism of type $(2,2,0)$
 \begin{equation*}
  \iota_{Y,Y'} :  (\mathcal{S}(Y'),\circ,+) = \langle \phi_x\rangle_{x \in Y'} \xhookrightarrow{} \langle \phi_x\rangle_{x \in Y} = (\mathcal{S}(Y),\circ,+)
 \end{equation*}
 
 \noindent of subalgebras of $(\mathcal{S}(X),\circ,+)$.
 We note that the homomorphisms $\iota_{Y,Y'}$ are compatible with the partial order on $\mathfrak{X}$, i.e.
 \begin{equation*}
 \iota_{Y,Y''}=\iota_{Y,Y'}\circ \iota_{Y',Y''}
 \end{equation*}

 \noindent for any $Y'' \subset Y' \subset Y$.
 Moreover, the homomorphisms $\iota_{Y,Y'}$ descend to homomorphisms of semigroups
 \begin{equation*}
  \iota^{\circ}_{Y,Y'} :  (\mathcal{S}(Y'),\circ) \xhookrightarrow{} (\mathcal{S}(Y),\circ)
 \end{equation*}
 
 \noindent and to homomorphisms of commutative semigroups of idempotents with zero
 \begin{equation*}
  \iota^+_{Y,Y'} :  (\mathcal{S}(Y'),+,0) \xhookrightarrow{} (\mathcal{S}(Y),+,0)
 \end{equation*}
 
 \noindent for $Y' \subset Y \subset X$.
 As set-maps, $\iota_{Y,Y'}$, $\iota^{\circ}_{Y,Y'}$ and $\iota^+_{Y,Y'}$ are identical and therefore we denote all of them by 
 \begin{equation*}
  \iota_{Y,Y'} : \mathcal{S}(Y') \xhookrightarrow{} \mathcal{S}(Y)
 \end{equation*}
 
 \noindent in agreement with definition \ref{def:alg_hom}, when the algebraic type is clear from the context.

\end{rmk}

\subsection{Semigroup extensions}

The following definitions are used in the last part of section 2 of the main text.

\begin{defn}
 A {\it short exact sequence of semigroups with zero} $(\mathcal{S},\circ,0)$, $(\mathcal{S}',\circ,0)$ and $(\mathcal{S}'',\circ,0)$ is an injective homomorphism $\iota: \mathcal{S}' \rightarrow \mathcal{S}$ and a surjective homomorhism $\pi: \mathcal{S} \rightarrow \mathcal{S}''$ such that $\pi(\iota(a)) = 0$ for all $a\in \mathcal{S}'$.
 A short exact sequence is represented as
 \begin{equation*}
 0 \rightarrow \mathcal{S}' \xrightarrow{\iota} \mathcal{S} \xrightarrow{\pi} \mathcal{S}'' \rightarrow 0.
 \end{equation*}
 
\end{defn}

\begin{defn}
  An {\it extension of a semigroup with zero} $(\mathcal{S}'',\circ,0)$ by $(\mathcal{S}',\circ,0)$ is a semigroup $(\mathcal{S},\circ,0)$ that fits into a short exact sequence $0 \rightarrow \mathcal{S}' \xrightarrow{\iota} \mathcal{S} \xrightarrow{\pi} \mathcal{S}'' \rightarrow 0$.
\end{defn}

\begin{defn}
 A short exact sequence of semigroups with zero 
 \begin{equation*}
 0 \rightarrow \mathcal{S}' \xrightarrow{\iota} \mathcal{S} \xrightarrow{\pi} \mathcal{S}'' \rightarrow 0
 \end{equation*}
 
 \noindent {\it splits}, if there is a homomorphism $s: \mathcal{S}'' \rightarrow \mathcal{S}$ of semigroups with zero, such that $\pi \circ s = \text{id}|_{\mathcal{S}''}$.
 Such an $s$ is called a {\it section} of $\pi$.
\end{defn}

\begin{rmk}
 If the sequence 
 \begin{equation*}
 0 \rightarrow \mathcal{S}' \xrightarrow{\iota} \mathcal{S} \xrightarrow{\pi} \mathcal{S}'' \rightarrow 0
 \end{equation*}
 
 \noindent as above splits, then $\mathcal{S}$ is isomorphic to the direct sum $\mathcal{S}' \oplus \mathcal{S}''$.
 In this case, we say that the extension $\mathcal{S}$ is {\it trivial}.
\end{rmk}

\section{Congruences} \label{sec:cong}

In this section, congruences on algebras and the resulting quotient algebras are defined and it is shown explicitly how the operations descend to the quotient.
For semigroups, Rees quotient semigroups are introduced as a specific example.
Finally, the lattice of congruences is discussed.\\

We begin with preliminary definitions.
The two following definitions are well-known, however, we introduce an equivalence relation on a set $\mathcal{A}$ as a subset of $\mathcal{A} \times \mathcal{A}$ in addition to the usual viewpoint as a partition of $\mathcal{A}$.

\begin{defn}
 Let $\mathcal{A}$ be a set.
 A {\it relation} $\rho$ on $\mathcal{A}$ is a subset of $\mathcal{A} \times \mathcal{A}$
 \begin{equation*}
  \rho \subset \mathcal{A} \times \mathcal{A}.
 \end{equation*}
 
 \noindent If $(a,b) \in \rho$, we say that $a$ and $b$ are related via $\rho$ and write
 \begin{equation*}
  a \rho b.
 \end{equation*}
 
\end{defn}

\begin{defn}
 Let $\mathcal{A}$ be a set.
 An {\it equivalence relation} $\rho$ on $\mathcal{A}$ is a relation that is reflexive, symmetric and transitive, i.e.
 \begin{align*}
  &a \rho a \text{ for all }a \in \mathcal{A}, \\
  &a \rho b \Rightarrow b \rho a \text{ for all }a,b \in \mathcal{A},\\
  &a \rho b \wedge b \rho c \Rightarrow a \rho c \text{ for all }a,b,c \in \mathcal{A}.
 \end{align*}
 
 \noindent Equivalently, $\rho$ can be identified with a partition of the set $\mathcal{A}$, i.e.
 \begin{equation*}
  \mathcal{A} = \coprod_{i \in I}\mathcal{A}_i.
 \end{equation*}
 
 \noindent Thereby, each $a \in \mathcal{A}$ is contained in exactly one coset $\mathcal{A}_i$, which contains all elements $b \in \mathcal{A}$ such that $a\rho b$ \
 and only those.
 The $\mathcal{A}_i$ are called {\it equivalence classes} or {\it cosets}.
 We denote the set of equivalence classes of $\mathcal{A}$ as $\mathcal{A}/\rho$.
 For any $a \in \mathcal{A}$, the unique set $\mathcal{A}_i$ containing $a$ is called the {\it equivalence class of $a$}.
 We denote the equivalence class of $a$ by $a\rho$.
 Moreover, for any equivalence class $\mathcal{A}_i$, any element $a \in \mathcal{A}_i$ is called {\it coset representative}.
\end{defn}

\begin{defn}
 Let $\mathfrak{A} = (\mathcal{A};F)$ be an algebra of algebraic type $\tau=(\mathcal{O},\alpha)$.
 A {\it congruence} $\rho$ on $\mathfrak{A}$ is an equivalence relation on $\mathcal{A}$ that is compatible with the algebraic operations of $\mathfrak{A}$, i.e. for all $f \in \mathcal{O}$ and all $a_1,...,a_{\alpha(f)},b_1,...,b_{\alpha(f)} \in \mathcal{A}$ the implication
 \begin{equation} \label{eq:cong}
  a_1 \rho b_1 \wedge a_2 \rho b_2 \wedge ... \wedge a_{\alpha(f)} \rho b_{\alpha(f)} \Rightarrow F(f)(a_1,...,a_{\alpha(f)}) \rho F(f)(b_1,...,b_{\alpha(f)}) 
 \end{equation}

 \noindent holds.
\end{defn}

We have the following {\bf main lemma}.

\begin{lem}
 Let $\mathfrak{A} = (\mathcal{A};F)$ be an algebra of algebraic type $\tau=(\mathcal{O},\alpha)$ and let $\rho$ be a congruence on $\mathfrak{A}$.
 Then the operations $F$ naturally descend to the set $\mathcal{A}/\rho$ as
 \begin{equation*}
  (F/\rho)(f)(a_1 \rho,...,a_{\alpha(f)}\rho) := F(f)(a_1,...,a_{\alpha(f)})\rho,
 \end{equation*}

 \noindent which is independent of choice of coset representatives $a_1,...,a_{\alpha(f)}$ by relation \ref{eq:cong} and therefore well-defined.
 Thus $\mathfrak{A}/\rho = (\mathcal{A}/\rho;F/\rho)$ is an algebra of type $\tau$.
 $\mathfrak{A}/\rho$ is called the {\it quotient algebra of $\mathfrak{A}$ by $\rho$} or just the {\it quotient of $\mathfrak{A}$ by $\rho$}.
\end{lem}

In the rest of this section, we discuss a specific example of congruences on semigroups. 
This requires a preliminary definition.

\begin{defn}
Let $(\mathcal{S},\circ)$ be a semigroup.
An {\it ideal} $\mathcal{I}$ is a proper subset of $\mathcal{S}$ such that
\begin{equation*}
 \mathcal{S}\mathcal{I} \cup \mathcal{I}\mathcal{S} \subset \mathcal{I},
\end{equation*}

\noindent where the notation
\begin{equation} \label{eq:prod}
 \mathcal{A}\mathcal{B} = \{a \circ b| a \in \mathcal{A}, b \in \mathcal{B}\}
\end{equation}

\noindent for $\mathcal{A},\mathcal{B} \subset \mathcal{S}$ is used.
\end{defn}

\begin{defn} \label{def:Rees}
Let $(\mathcal{S}(X),\circ)$ be a semigroup and $\mathcal{I} \subset \mathcal{S}$ an ideal.
Define a congruence $\rho_{\mathcal{I}}$ via
\begin{equation} \label{eq:Rees}
 \rho_{\mathcal{I}} = \{(a,b)|a,b \in \mathcal{I}\} \cup \{(c,c)|c \in \mathcal{S} \}.
\end{equation}

\noindent The {\it Rees factor semigroup} is the quotient semigroup $\mathcal{S}/\rho_{\mathcal{I}}$.
It is denoted by $\mathcal{S}/\mathcal{I}$.
\end{defn}

\begin{rmk}
 Note that in the Rees factor semigroup  $\mathcal{S}/\mathcal{I}$, all elements of $\mathcal{I}$ are identified, i.e. they are in the same equivalence class, and all elements of $\mathcal{S}(X) \setminus \mathcal{I}$ remain in their own separate equivalence classes.
\end{rmk}

\begin{rmk}
It is important to mention that the language congruences does not lead to any new features for groups (and therefore rings, modules and algebras in commutative algebra), but is crucial in universal algebra, e.g. already for semigroups.
Indeed, for any group $\mathcal{G}$, a congruence $\rho$ is uniquely determined by a normal subgroup $\mathcal{N} < \mathcal{G}$ via $a\rho b \Leftrightarrow ab^{-1} \in \mathcal{N}$ and each normal subgroup uniquely corresponds to a congruence as the kernel of the projection $\mathcal{G} \rightarrow \mathcal{G}/\rho$.
Thus, the study of congruences is reduced to the study of normal subgroups.
However, for semigroups, the congruences are not always determined by subsemigroups in the same manner as for groups.
For example, congruences on finite semigroups can yield congruence classes of different sizes.
This is the case for all Rees quotients of a finite semigroup $\mathcal{S}$ by a proper ideal $\mathcal{I} \subset \mathcal{S}$.
Hereby, all elements of $\mathcal{S} \setminus \mathcal{I}$ form separate classes, whereas all elements of $\mathcal{I}$ belong to the same class.
In contrast, in quotients of groups $\mathcal{G}/\mathcal{N}$ all congruence classes are in bijection with the respective normal subgroup $\mathcal{N}$ and thus necessarily have the same size.
\end{rmk}

For the remainder of this section, we fix an algebra $\mathfrak{A} = (\mathcal{A};F)$.
Let $\mathrm{Con}(\mathfrak{A})$ be the set of all congruences on $\mathfrak{A}$.
Each congruence $\rho$ is a subset of $\mathcal{A} \times \mathcal{A}$ and thus $\mathrm{Con}(\mathfrak{A})$ is partially ordered by inclusion of sets, i.e.
\begin{equation*}
 \rho  \leq \rho' \Leftrightarrow \rho \subset \rho'
\end{equation*}

\noindent for any $\rho,\rho'\in \mathrm{Con}(\mathfrak{A})$.
As the intersection of congruences is still a congruence, $(\mathrm{Con}(\mathfrak{A}),\leq)$ admits arbitrary infima.
Because a supremum of a subset is just the infimum of the set of upper bounds, $(\mathrm{Con}(\mathfrak{A}),\leq)$ admits arbitrary suprema and is therefore a lattice by proposition \ref{prop:cong}.\\

Let $\rho,\rho' \in \mathrm{Con}(\mathfrak{A})$ be two arbitrary congruences. 
If $\rho < \rho'$, we say that $\rho$ is {\it finer} than $\rho'$ and, vice versa, that $\rho'$ is {\it coarser} than $\rho$.
The lattice of congruences has a maximal element $\mathbbm{1} = \mathcal{A} \times \mathcal{A}$ and a minimal element $\Delta = \{(a,a)|a\in \mathcal{A}\} \subset \mathcal{A} \times \mathcal{A}$.\\

We have the following lemma.

\begin{lem}
Let $\rho$ be a congruence of $\mathfrak{A}$.
There is a one-to-one correspondence between the congruences $\rho'$ of $\mathfrak{A}$ coarse than $\rho$ and the congruences of $\mathfrak{A}/\rho$:
\begin{equation*}
\mathrm{Con}(\mathfrak{A}/\rho) \xleftrightarrow{\text{1-to-1}} \{\rho' \in \mathrm{Con}(\mathfrak{A}) \text{ such that } \rho \leq \rho'\}.
\end{equation*}

\noindent Moreover, let $\rho' \in \mathrm{Con}(\mathfrak{A})$ be such that $\rho \leq \rho'$ and let $\overline{\rho}'$ be the corresponding congruence in $\mathrm{Con}(\mathfrak{A}/\rho)$.
Then there is a natural algebra isomorphism
\begin{equation*}
 \mathfrak{A}/\rho' \simeq (\mathfrak{A}/\rho)/\overline{\rho}'.
\end{equation*}
\end{lem}

\begin{rmk}
 This lemma is central in the coarse-graining procedure via congruences as it ensures that the final result of subsequent coarse-graining procedures by increasingly coarser congruences
 \begin{equation*}
  \rho_1 \leq \rho_2 \leq ... \leq \rho_n 
 \end{equation*}
 
 \noindent is independent of the sequence $\{\rho_i\}_{i=1}^n$, but only depends on the final congruence $\rho_n$.
 Moreover, with the notations as in the lemma, it ensures that the lattice $\mathrm{Con}(\mathfrak{A}/\rho)$ includes all possible coarse-graining procedures that are given by the lattice $\mathrm{Con}(\mathfrak{A})$ after fixing the congruence $\rho$.
 That means that the sequence $\rho_1 \leq \rho_2 \leq ... \leq \rho_n$ can be either selected in $\mathrm{Con}(\mathfrak{A})$ at once or constructed step by step  by iteratively choosing the congruence in $\mathrm{Con}(\mathfrak{A}/\rho_i)$ corresponding to $\rho_{i+1}$ after coarse-graining by $\rho_i$.
\end{rmk}

\section{Congruences on $(\mathcal{S}(X),\circ)$} \label{sec:congSemigroupModel}

As a supplement to the main text, where we put the focus on congruences on $(\mathcal{S}(X),+)$, we discuss a class of congruences on $(\mathcal{S}(X),\circ)$ that are similar in spirit to the congruences $\rho^k$ from the main text.
Here, we present these congruences to illustrate the flexibility of our approach to coarse-graining.
We write $\mathcal{S}(X)$ for $(\mathcal{S}(X),\circ)$ throughout this section as the operation is understood to be $\circ$.\\

Consider the chain of ideals 
\begin{equation*}
\mathcal{S}(X) \supsetneq \mathcal{S}(X)^2 \supsetneq ... \supsetneq \mathcal{S}(X)^N = \mathcal{S}(X)^{N+1},
\end{equation*}

\noindent where the notation \ref{eq:prod} is used.
The sequence stabilizes for some $N \in \mathbb{N}$ due to the finiteness of $\mathcal{S}(X)$.
The powers $\mathcal{S}(X)^n$ are proper ideals of $\mathcal{S}(X)$ for $2 \leq n \leq N$ and give rise to congruences $\theta_{\mathcal{S}(X)^n}$ via the expression \ref{eq:Rees}.
For notational convenience, we will write $\theta^n := \theta_{\mathcal{S}(X)^n}$.
The quotient semigroups $\mathcal{S}(X) / \theta^n$ can be interpreted via the complexity of function defined as follows.

\begin{defn}
Let $\phi$ be some function in the semigroup model $\mathcal{S}(X)$.
$\phi$ has {\it complexity} $n$ if there exists some $n \in \mathbb{N}$ with $1 \leq n \leq N$ such that
\begin{equation*}
\phi \in \mathcal{S}(X)^n \setminus \mathcal{S}(X)^{n+1}.
\end{equation*}

\noindent Constant functions (including 0) have complexity $\infty$.
The complexity of $\phi$ is denoted by $comp(\phi)$.

\end{defn}

The complexity $comp(\phi)$ of a function $\phi$ determines whether the function can be decomposed into a product of at most $comp(\phi)$ functions.
For example, a non-constant function $\phi_x$ of a chemical $x \in X$ has complexity $1$, because it cannot be further decomposed.
By remark 7 from the main text, functions correspond to reaction pathways within the CRS.
Intuitively, $comp(\phi)$ gives the length of the shortest pathway described by $\phi$.
By definition, any two functions $\phi, \psi \in \mathcal{S}(X)$ satisfy the inequality
\begin{equation} \label{eq:comp}
comp(\phi) + comp(\psi) \leq comp(\phi \circ \psi).
\end{equation}

\begin{ex} \label{ex:comp}

The CRS shown in figure \ref{fig:comp} demonstrates that the inequality \ref{eq:comp} can be strict.
The functions $\phi = \phi_{x_2} \circ \phi_{x_1} + \phi_{y_1}$ and $\psi = \phi_{x_3} + \phi_{y_3} \circ \phi_{y_2}$ have complexity 1.
Their composition can be written as $(\phi_{x_3} + \phi_{y_3}) \circ (\phi_{x_2} + \phi_{y_2}) \circ (\phi_{x_1} + \phi_{y_1})$ and thus has complexity 3.

\begin{figure}[h]
  \centering
  \includegraphics[scale=0.30]{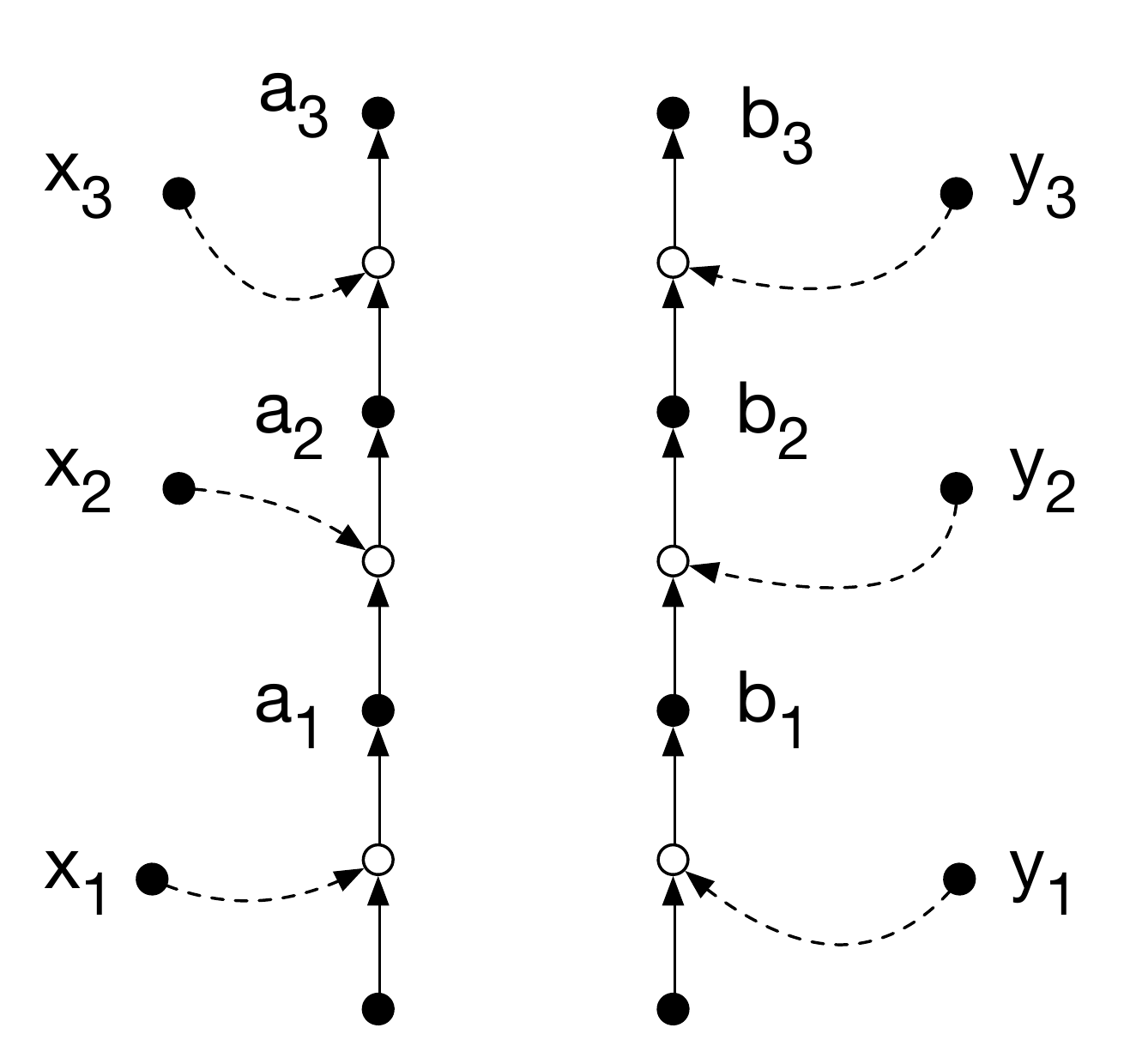}
  \caption[Example of functions $\phi, \psi \in \mathcal{S}_F$ with $comp(\phi) + comp(\psi) < comp(\phi \circ \psi)$]{
  The functions $\phi =   (\phi_{x_2} \circ \phi_{x_1}) + \phi_{y_1}$ and $\psi = \phi_{x_3} + (\phi_{y_3} \circ \phi_{y_2})$ have complexity 1, but their composition has complexity 3.
  }
  \label{fig:comp}
\end{figure}

\end{ex}

The quotient semigroups $\mathcal{S}(X) / \theta^n$ are the semigroups of functions of complexity at most $n$, i.e. the functions with complexity lower than $n$ are all in separate congruence classes and the functions with complexity greater or equal to $n$ are in the congruence class of 0.\\

The composition of two functions $\phi, \psi \in \mathcal{S}(X) / \theta^n$ with $comp(\phi),comp(\psi) < n$ gives $\phi \circ \psi$ if $comp(\phi \circ \psi) < n$ and zero otherwise.
Thus, the quotient $\mathcal{S}(X) / \theta^n$ naturally injects into $\mathcal{S}(X) / \theta^{n+1}$ for $2 \leq n \leq N-1$ as a set
\begin{equation*}
 \iota_n: \mathcal{S}(X) / \theta^n \xhookrightarrow{} \mathcal{S}(X) / \theta^{n+1}.
\end{equation*}

\noindent However, this is not a semigroup homomorphism.
Furthermore, the congruences $\theta^n$ are totally ordered by inclusion as
\begin{equation*}
 \theta^N \supset \theta^{N-1} \supset ... \supset \theta^2
\end{equation*}

\noindent and give rise to projections
\begin{equation*}
 \pi_n: \mathcal{S}(X) / \theta^{n+1} \twoheadrightarrow \mathcal{S}(X) / \theta^n ,
\end{equation*}

\noindent where the $\pi_n$ are semigroup homomorphisms.\\

A biological interpretation of the quotients $\mathcal{S}(X) / \theta^n$ is as follows: 
They capture the local structure of the CRS of ``size at most $n$'', i.e. within the quotient $\mathcal{S}(X) / \theta^n $ it is only possible to see those functions that contain reaction pathways of length smaller than $n$.
It is possible to compose the functions as usual, but as soon as the compositions gain a complexity larger than $n$, the functions vanish, i.e. one is restricted to interactions within ``local patches'' of limited size.
Returning to the idea of relating congruences to coarse-graining schemes, the $\theta^n$ describe a rather unusual coarse-graining of the system:
Lumping together functions of large complexity can be thought of lumping together ``the environment'' and retaining the local structure.
However, the coarse-graining via the $\theta^n$ does not fix a given subnetwork and then integrates out its environment, but preserves all the local patches.
It is possible to combine functions in $\mathcal{S}(X) / \theta^n$ that seemingly live on different patches.
The injections $\iota_n$ are inclusions of patches of size $n$ into patches of size $n+1$ and the projections $\pi_n$ lose information about functions with complexity $n+1$ and thus correspond to a reduction to smaller patches.
This interpretation of the quotients $\mathcal{S}(X)/\theta^n$ as a coarse-graining of the environment is illustrated in figure \ref{fig:EnvCG}.

\begin{figure}[h]
  \centering
  \includegraphics[width=.8\linewidth]{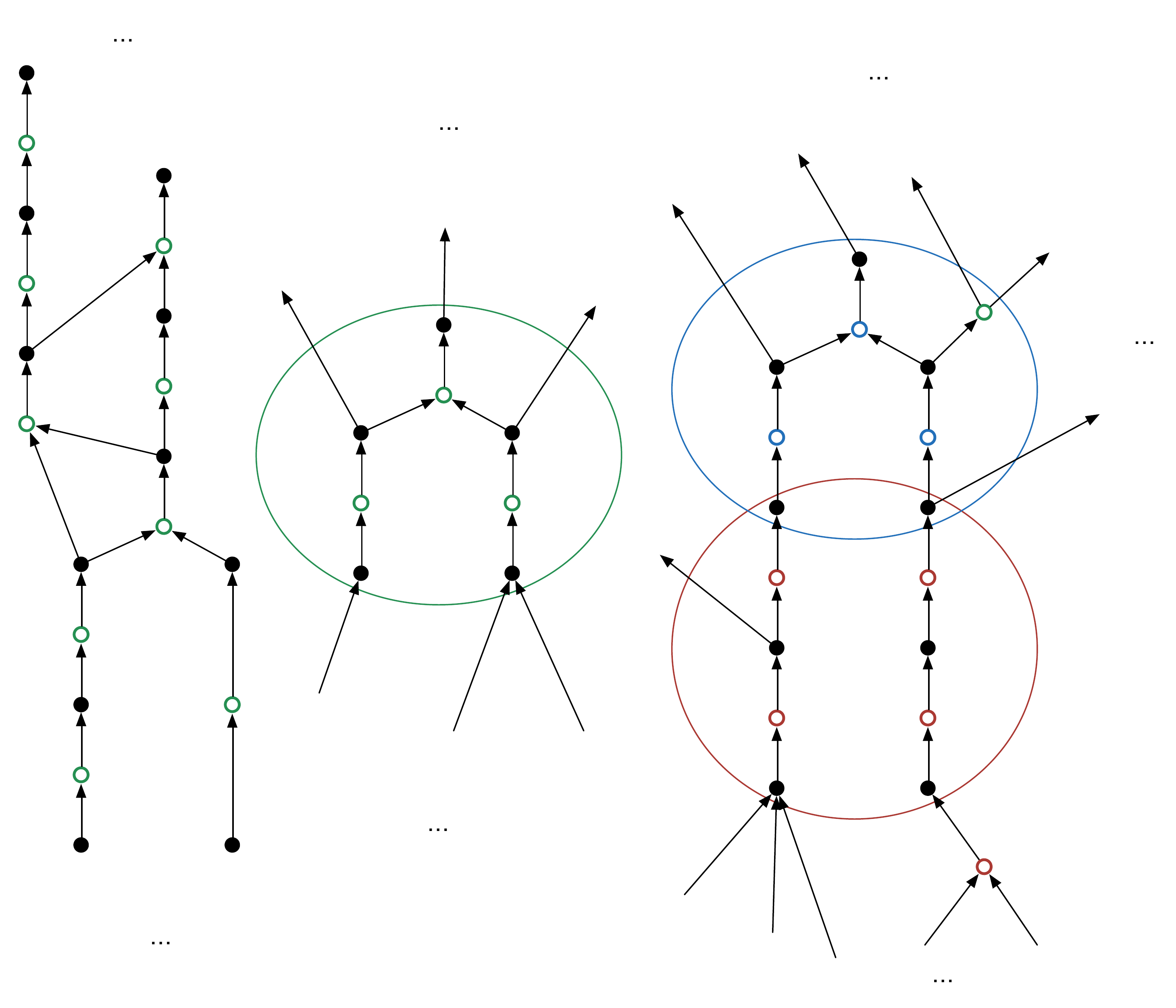}
  \caption[Illustration of coarse-graining of the environment via $\theta^{n}$]{Illustration of coarse-graining of the environment via the congruence $\theta^{3}$.
  The figure shows three functions $\phi_{\text{green}}$, $\phi_{\text{blue}}$, $\phi_{\text{red}}$ colored in green, blue and red via the representation of elements in $\mathcal{S}(X)$ as pathways in the CRS.
  The circles indicate the local patches of complexity at most 2.
  Each of the functions has a local structure of complexity 2 lying in the respective circles.
  The functions $\phi_{\text{green}}$, $\phi_{\text{blue}}$, $\phi_{\text{red}}$  are nonzero in $\mathcal{S}(X)/ \theta^3$.
  The composition $\phi_{\text{green}} \circ \phi_{\text{blue}}$ gives the function in the blue patch.
  It has complexity $\leq 2$ as well.
  The composition $\phi_{\text{blue}} \circ \phi_{\text{red}}$ has complexity 4 and equals zero in $\mathcal{S}(X) / \theta^3$.
  }
  \label{fig:EnvCG}
\end{figure}

\section{Coarse-graining by the congruences ${\bf \rho_n}$}

As an illustrative example, we discuss the geometric object
\begin{equation}\label{eq:CGgeometry}
  \begin{tikzcd}
    \mathcal{S}/\rho  \arrow[swap]{d}{\pi} & \\
     \mathfrak{X}/\rho_X, \arrow[swap,bend right]{u}{f^{\rho}} &
  \end{tikzcd}
\end{equation}

\noindent which results from the coarse-graining procedure in section 3 of the main text by the congruences $\rho_n$ (defined in section 2 of the main text).
Let the dynamics $f:\mathfrak{X}\rightarrow \mathcal{S}$ be given by $f_Y = \Phi_Y$ and $\rho_n$ be the congruence given by 
\begin{equation*} 
 \phi \rho_n \psi \Leftrightarrow len(\phi) \leq n \textrm{ and }len(\psi) \leq n.
\end{equation*}

\noindent Write $\rho := \rho_n$ and denote the respective equivalence relation resulting from $\rho_n$ via equation [11] in the main text by $\rho_X^{\text{pre}}$ and the equivalence relation on $\mathfrak{X}$ resulting from definition 10 in the main text by $\rho_X$.
Moreover, define $len(\mathcal{S}(Y)) := len(\Phi_Y)$ and $len(Y) := len(\Phi_Y)$ for any $Y \subset X$.\\

Recall that $\rho_n$ identifies all functions of length less or equal to $n$ with zero.
Therefore, all $Y \subset X$ with $len(Y) \leq n$ are equivalent to the empty set by the relation $\rho_X^{\text{pre}}$.
This means that a certain lower part of the lattice $\mathfrak{X}$ is reduced to a single element $\emptyset\rho_X^{\text{pre}}\in \mathfrak{X}/\rho_X^{\text{pre}}$ with only the zero function in its algebra.
These are all the subsets of $X$ that support only functions of low length.
All other elements of $\mathfrak{X}$ are in separate congruence classes.
Moreover, all algebras $\mathcal{S}(Y)$ with $len(Y) > n$ only have functions longer than $n$, which are each their own equivalence class, and the zero function.\\

To construct the relation $\rho_X$ and the dynamics $f^{\rho}$, it is necessary to discuss 3 different cases:\\

{\bf 1)} If for all $Y,Y' \in \emptyset\rho_X^{\text{pre}}$ and any $\phi \in \mathcal{S}(Y')$, we have $\phi(Y) \in \emptyset\rho_X^{\text{pre}}$, then by definition 10 from the main text, $f^{\rho}$ is well-defined as $f^{\rho}(\emptyset\rho_X^{\text{pre}}) = \emptyset\rho_X^{\text{pre}}$.
For any set $Z \in \mathfrak{X} \setminus \emptyset\rho_X^{\text{pre}}X$, there are no changes from the original dynamics as both $Z$ and $\Phi_Z$ are in separate equivalence classes and $f^{\rho}(Z) = \Phi_Z(Z)$ is well-defined.
In this case, $\rho_X' = \rho_X$ and the coarse-graining of functions of low length leads to the contraction of all $Y \subset X$ supporting only such functions and retains all other sets and functions in a manner consistent with the original dynamics.\\

{\bf 2)} If there are $Y,Y' \in \emptyset\rho_X^{\text{pre}}$ such that $\Phi_Y(Y') \in \mathfrak{X} \setminus \emptyset\rho_X^{\text{pre}}$, then by definition 10 from the main text $\Phi_Y(Y')$ must be in $\emptyset \rho_X$.
Iteratively adding all the sets $\Phi_Y(Y')$ for $Y,Y' \in \emptyset\rho_X$ to $\emptyset\rho_X$ until $\emptyset\rho_X$ is stable under this operation implies that $f^{\rho}(\emptyset\rho_X) = \emptyset\rho_X$ is well-defined.
As in case 1), all $Z \in \mathfrak{X} \setminus \emptyset\rho_X$ and the corresponding $\Phi_Z$ are in separate equivalence classes and $f^{\rho}(Z) = \Phi_Z(Z)$ is well-defined.
This means that $\emptyset\rho_X$ contains all the sets that support only short functions and all sets that can be produced by short functions and the functions supported on the produced sets.\\

{\bf 3)} If, similar to case 2), $\emptyset\rho_X$ is the whole state space $\mathfrak{X}$, then successive combinations of functionality of length $n$ are enough to generate the whole network.
The coarse-grained geometrical model is trivial in this case.\\

This example demonstrates how congruences as simple as the $\rho_n$ highlight functional aspects of biochemical reaction networks.
For small $n$, it is to be expected that the geometrical coarse-graining procedure follows case 1) and leads with increasing $n$ via case 2) to case 3).
For real biological systems it is already interesting to study quantitative changes of this behavior in biological reaction networks and compare them to random networks.
We note that the congruences $\rho_n$ are rather coarse that in many cases leads eventually to the complete contraction of the phase space..
However, it is as well possible to construct congruences that are finer than $\rho_n$ and pay attention to functional modularity of the network.
In such cases, the procedure yields functional partitions highlighting the interplay of the respective modules.
This will be the topic of forthcoming work.

\newpage

\section{Analogy between classical and algebraic models}

\begin{table}[hbpt]
\centering
\caption{
Analogy between classical models describing the dynamics of a chemical reaction network by a set of ordinary differential equations and the algebraic models.
Note that $f_Y$ denotes the section $f$ at $Y$, which is a function $\mathfrak{X} \rightarrow \mathfrak{X}$ and $f_Y(Y)$ is its value at $Y$.
We also use the notations $N = |X|$ and $\mathfrak{X}=\{0,1\}^X$.
}
\begin{tabular}{p{3cm} | p{3.5cm} p{3.5cm}}
 & Classical & Algebraic \\
\hline
State space & $\mathbb{R}^{N}_{\geq 0}$ & $\mathfrak{X}$  \\
``Geometry`` of state space & Euclidean topology & Partial order by inclusion  \\
Space of functions & $T\mathbb{R}^{N}_{\geq 0}=\coprod_{x \in \mathbb{R}^N_{\geq 0}} T_x \mathbb{R}^N_{\geq 0}$ & $\mathcal{S}=\coprod_{Y \in \mathfrak{X}} \mathcal{S}(Y)$ \\
Attachment of functions to state space & $T\mathbb{R}^{N}_{\geq 0} \xrightarrow{\pi}\mathbb{R}^{N}_{\geq 0} $ & $\mathcal{S} \xrightarrow{\pi} \mathfrak{X}$ \\
Topology of attachment & Natural topology & Compatibility of partial orders \\
Dynamics & Smooth section $\mathbb{R}^{N}_{\geq 0} \xrightarrow{f_{\text{class}}}T\mathbb{R}^{N}_{\geq 0} $ & Section $\mathfrak{X}\xrightarrow{f}  \mathcal{S}$ compatible with partial order \\
Trajectory & Integration of $dx/dt = f_{\text{class}}(x)$ with initial condition $x_0 \in \mathbb{R}^{N}_{\geq 0}$ & Iteration of $Y \mapsto f_Y(Y)$ with initial condition $Y_0 \subset X$ \\
Parametrization of a trajectory & $(\mathbb{R}_{\geq 0},+)$ & $(\mathbb{N},+)$ \\ 
\hline
\end{tabular}
\label{table:geometry}
\end{table}

%%% Add this line AFTER all your figures and tables

\iffalse

\movie{Type caption for the movie here.}

\movie{Type caption for the other movie here. Adding longer text to show what happens, to decide on alignment and/or indentations.}

\movie{A third movie, just for kicks.}

\dataset{dataset_one.txt}{Type or paste caption here.}

\dataset{dataset_two.txt}{Type or paste caption here. Adding longer text to show what happens, to decide on alignment and/or indentations for multi-line or paragraph captions.}

\fi

\newpage

\bibliographystyle{spphys}       % APS-like style for physics
\bibliography{literatur_SI}   % name your BibTeX data base

%\begin{thebibliography}{}